\documentclass[pra,onecolumn,superscriptaddress,nofootinbib]{revtex4}
\usepackage[a4paper, left=2cm, right=2cm, top=2cm, bottom=2cm]{geometry}

\usepackage[english]{babel}
\usepackage[utf8]{inputenc}
\usepackage[T1]{fontenc}
\usepackage{amsthm}
\newtheorem{thm}{Theorem}[section]
\newtheorem{cor}{Corollary}[section]
\newtheorem{lem}{Lemma}[section]

\newtheorem{quest}{Question}[thm]
\newtheorem{fact}{Fact}[section]
\newtheorem{prop}{Proposition}[section]

\newtheorem{defn}{Definition}[section]
\usepackage{comment}

\usepackage{amsmath}
\usepackage{graphicx}
\usepackage{mathtools}
\mathtoolsset{showonlyrefs=true}
\usepackage[colorinlistoftodos]{todonotes}
\usepackage[colorlinks=true, allcolors=blue]{hyperref}
\usepackage{bbm,microtype,mathrsfs,amsmath,amssymb,color,amsthm,graphicx,bm} 
\usepackage{tikz-cd}
\usepackage{xcolor}
\usepackage{braket}
\usepackage{thm-restate}

\newcommand{\CD}{\mathcal{D}}
\newcommand{\vCD}{\bm{\mathcal{D}}}
\newcommand{\tvCD}{\tilde{\bm{\mathcal{D}}}}

\newcommand{\BE }{\mathbb{E}}

\newcommand{\CN}{\mathcal{N}}

\newcommand{\CO}{\mathcal{O}}

\newcommand{\BR}{\mathbb{R}}

\newcommand{\vA}{\bm{A}}
\newcommand{\tvA}{\tilde{\bm{A}}}

\newcommand{\vB}{\bm{B}}
\newcommand{\vC}{\bm{C}}
\newcommand{\vD}{\bm{D}}

\newcommand{\vE}{\bm{E}}

\newcommand{\vG}{\bm{G}}
\newcommand{\vH}{\bm{H}}
\newcommand{\tvH}{\tilde{\bm{H}}}
\newcommand{\vI}{\mathbf{I}}

\newcommand{\vM}{\bm{M}}

\newcommand{\vO}{\bm{O}}
\newcommand{\vP}{\bm{P}}
\newcommand{\tvP}{\tilde{\bm{P}}}
\newcommand{\vQ}{\bm{Q}}
\newcommand{\tvQ}{\tilde{\bm{Q}}}
\newcommand{\vR}{\bm{R}}
\newcommand{\tvR}{\tilde{\bm{R}}}

\newcommand{\vS}{\bm{S}}
\newcommand{\tvS}{\tilde{\bm{S}}}

\newcommand{\vU}{\bm{U}}
\newcommand{\tvU}{\tilde{\bm{U}}}
\newcommand{\vV}{\bm{V}}

\newcommand{\vx}{\bm{x}}
\newcommand{\vX}{\bm{X}}
\newcommand{\vy}{\bm{y}}
\newcommand{\vY}{\bm{Y}}
\newcommand{\vZ}{\bm{Z}}
\newcommand{\vsigma}{\bm{\sigma}}
\newcommand{\vrho}{\bm{\rho}}

\renewcommand{\L}{\left}
\newcommand{\R}{\right}

\newcommand{\bomega}{\bar{\omega}}

\newcommand{\dagg}{\dagger}

\newcommand{\AC}[1]{ {\color{blue} [(Anthony:) #1]}}

\newcommand*{\tr}{\operatorname{Tr}}
\newcommand*{\btr}{\operatorname{\overline{Tr}}}
\newcommand*{\poly}{\mathrm{Poly}}

\newcommand{\vertiiismall}[1]{{\vert\kern-0.25ex\vert\kern-0.25ex\vert #1 
    \vert\kern-0.25ex\vert\kern-0.25ex\vert}}

\newcommand{\vertiii}[1]{{\left\vert\kern-0.25ex\left\vert\kern-0.25ex\left\vert #1 
    \right\vert\kern-0.25ex\right\vert\kern-0.25ex\right\vert}}
\newcommand{\vertiiiNoLR}[1]{{\bigg\vert\kern-0.25ex\bigg\vert\kern-0.25ex\bigg\vert #1 
    \bigg\vert\kern-0.25ex\bigg\vert\kern-0.25ex\bigg\vert}}
\newcommand{\lvertiii}{\bigg\vert\kern-0.25ex\bigg\vert\kern-0.25ex\bigg\vert }
\newcommand{\rvertiii}{\bigg\vert\kern-0.25ex\bigg\vert\kern-0.25ex\bigg\vert }
\newcommand{\norm}[1]{\Vert {#1} \Vert}
\newcommand{\lnorm}[1]{\left\Vert {#1} \right\Vert}

\newcommand{\normp}[2]{\norm{#1}_{#2}}

\newcommand{\labs}[1]{\left\vert {#1} \right\vert}
\newcommand{\abs}[1]{\vert {#1}\vert}
\newcommand{\e}{\mathrm{e}}
\newcommand{\iunit}{\mathrm{i}}
\newcommand{\ri}{\mathrm{i}}

\newcommand{\indicator}{\mathbbm{1}}
\newcommand{\Expect}{\operatorname{\mathbb{E}}}

\newcommand{\diff}{\mathrm{d}}
\newcommand{\idiff}{\,\mathrm{d}}

\begin{document}
\title{ 
Sparse random Hamiltonians are quantumly easy
}

\author{Chi-Fang (Anthony) Chen}
\affiliation{Institute for Quantum Information and Matter,
California Institute of Technology, Pasadena, CA, USA}
\affiliation{AWS Center for Quantum Computing, Pasadena, CA}

\author{Alexander M. Dalzell}
\affiliation{AWS Center for Quantum Computing, Pasadena, CA}
\affiliation{California Institute of Technology, Pasadena, CA, USA}

\author{Mario Berta}
\affiliation{Institute for Quantum Information, RWTH Aachen University, Aachen, Germany}
\affiliation{Department of Computing, Imperial College London, London, UK}

\author{Fernando G.S.L. Brand\~ao}
\affiliation{AWS Center for Quantum Computing, Pasadena, CA}
\affiliation{Institute for Quantum Information and Matter,
California Institute of Technology, Pasadena, CA, USA}

\author{Joel A.~Tropp}
\affiliation{Department of Computing and Mathematical Sciences, Caltech, Pasadena, CA, USA}


\begin{abstract}
A candidate application for quantum computers is to simulate the low-temperature properties of quantum systems. 
For this task, there is a well-studied quantum algorithm that performs quantum phase estimation on an initial trial state that has a nonnegligible overlap with a low-energy state. 
However, it is notoriously hard to give theoretical guarantees that such a trial state can be prepared efficiently. Moreover, the heuristic proposals that are currently available, such as with adiabatic state preparation, appear insufficient in practical cases.



This paper shows that, for most random sparse Hamiltonians, the maximally mixed state is a sufficiently good trial state and phase estimation efficiently prepares states with energy arbitrarily close to the ground energy.
Furthermore, any low-energy state must have nonnegligible quantum circuit complexity, suggesting that low-energy states are classically nontrivial and phase estimation is the optimal method for preparing such states (up to polynomial factors).  These statements hold for two models of random Hamiltonians: (i) a sum of random signed Pauli strings and (ii) a random signed $d$-sparse Hamiltonian.
The main technical argument is based on some new results in nonasymptotic random matrix theory.   In particular, a refined concentration bound for the spectral density is required to obtain complexity guarantees for these random Hamiltonians.

\end{abstract}
\maketitle



\section{Introduction}


What are quantum computers good at? The earliest (and still most compelling) candidates are factoring~\cite{Factoring_Shor} and simulation of quantum systems~\cite{feynman1982simulating,lloyd1996universal,low2017optimal}. 
While Shor's celebrated quantum algorithm for factoring~\cite{Factoring_Shor} settles the quantum complexity of factoring, the complexity of quantum simulation at low energies has not been resolved.  Indeed, we know from complexity theory that the ground energy problem for general local Hamiltonians is QMA-hard in the worst case \cite{kitaev2002classical}, so we anticipate that preparing ground states is generally intractable, even for quantum computers.   This worst-case hardness persists for systems with additional physical constraints, including nearest-neighbor interaction in 1D \cite{aharonov2009power} and translation invariance \cite{gottesman2009quantum}.
Although these results send a pessimistic signal, they merely indicate that any \textit{proof} that the ground-energy problem
is quantumly easy must further constrain the class of Hamiltonians, or else it can apply only for typical instances. Indeed, one may construct random families of Hamiltonians (Section~\ref{sec:related_models}) in the hope that their \textit{average-case} complexity might be more favorable than the worst case.

Aside from complexity theory, the problem of preparing low-energy states arises in efforts to apply quantum computers to computational chemistry (for example, see~\cite{mcardle2020quantum}) and to condensed matter physics. To prepare a state of sufficiently low energy on a quantum computer, which can be used, e.g., for understanding chemical reaction pathways, a proposed quantum algorithm simply runs phase estimation on an initial trial state~\cite{THC_google,2021_Microsoft_catalysis,babbush2018low,chamberland2020building}. This method is efficient if the initial state has a nonnegligible overlap with a low-energy state.  Although the phase estimation part of the algorithm is well understood, we have an incomplete understanding of the time required to prepare a good initial state.  In fact, recent numerical tests~\cite{Isthere_22_lee} suggest that, for some chemical systems, easily preparable initial states may have exponentially small (in system size) overlap. Moreover, preparing states with good overlap using the adiabatic algorithm may take exponential time, significantly impacting the end-to-end performance of the proposed quantum algorithm. 

The search for tasks that are easy for quantum computers, in quantum chemistry or otherwise, is often implicitly a quest for \textit{quantum advantage}: quantum computers can be particularly helpful if the task is also classically hard. Unfortunately, proving classical hardness is challenging, and many once-promising candidates for classically hard problems have now been \textit{dequantized}.  For example, under certain classical access models, recent progress eliminates exponential quantum advantage in low-rank linear algebra tasks~\cite{Tang2018QuantumPC, Gilyn2018QuantuminspiredLS, Chia2019SamplingbasedSL,quantum_inspired_Tang_2019}. Still, hope remains that the Hamiltonian low-energy problem could provide a quantum advantage~\cite{Gharibian_2022}.
%
%
With these thoughts in mind, the guiding question of this work is the following.
\begin{quote}
\textit{``Is there a classically nontrivial Hamiltonian whose low-energy states are provably easy to prepare?''}\label{eq:guiding_question}
\end{quote}

Our work argues in the affirmative.
In particular, we will show that the textbook phase estimation method (discussed above) works well for preparing low-energy states of a typical random sparse Hamiltonians.  Meanwhile, the low-energy states must have a large quantum circuit complexity, so they are plausibly nontrivial for classical computers to simulate.

The paper is organized as follows.  First, we review relevant classes of Hamiltonians (Section~\ref{sec:related_models}) before presenting the main Hamiltonian model and the main results (Section~\ref{sec:main_results}). Our proof strategy (Section~\ref{sec:proof_ideas}) exploits tools from nonasymptotic random matrix theory; Section~\ref{sec:random-matrix} contains further details and context.  Last, we discuss the \textit{classical} complexity of the low-energy problem, and we lay out future research directions in the search for quantum advantage (Section~\ref{sec:q_advantage}). 


\subsection{Related Models}
\label{sec:related_models}
Before we give a statement of our main results, let us discuss how some familiar models fall short of answering our question. We focus on random ensembles where quantitative statements are available.
\begin{itemize}
    \item \textbf{The few-body Pauli models}~\cite{Erd_s_2014}. As a natural generalization of the classical spin glasses (e.g., the Sherrington--Kirkpatrick (SK) model~\cite{SK_model_75}), one replaces classical (commuting) constraints with noncommuting Pauli operators. A representative is the ensemble of Hamiltonians given by 
\begin{align}
    \vH := \sum_{i>j} g^x_{ij} \vsigma^x_i \vsigma^x_j + g^y_{ij}\vsigma^y_i \vsigma^y_j+ g^z_{ij}\vsigma^z_i \vsigma^z_j\,, \quad \text{where} \quad g^x_{ij},g^y_{ij}, g^z_{ij}\sim \text{i.i.d.~Gaussians},\label{eq:2bodyPauli}
\end{align}
and $\vsigma^x_i, \vsigma^y_i, \vsigma^z_i$ denote the Pauli operators on qubit $i$. Heuristically, this model exhibits spin glass behavior at low temperatures \cite{Swingle20_SK_SYK}, which suggests that finding low-energy states could be hard, even for quantum computers.\footnote{ If we regard spin glasses as NP-hard problems, then we do not expect that quantum computers can solve them efficiently.} 
In the high-temperature regime, this model becomes \textit{classically} easy: there exists an efficient algorithm that outputs a \textit{product state} approximating the operator norm of the Hamiltonian $\norm{\vH}$ to a constant ratio~\cite{Harrow2017extremaleigenvalues}. We do not know whether there is a temperature range where the state remains quantumly easy but classically hard, nor do we know how to attack this question.

\item \textbf{The Sachdev--Ye--Kitaev (SYK) models}~\cite{Sachdev_1993,Kitaev15} with fermionic degrees of freedom. For a representative, consider the four-body Hamiltonian
\begin{align}
    \vH_{SYK} := \sum_{i<j<k<\ell} g_{ijk\ell} \chi_i\chi_j \chi_k \chi_{\ell} \quad \text{where} \quad \chi_i\chi_j+\chi_j\chi_i = 2\delta_{ij} \quad \text{and}\quad g_{ijk\ell}\sim \text{i.i.d.~Gaussians}.
\end{align}
Nonrigorous arguments rooted in physics suggest that this model remains \textit{chaotic} (instead of a spin glass) at very low temperatures~\cite{Maldacena_rmkSYK, SYK_glass_2018_Gur_Ari,Swingle20_SK_SYK}.  If true, this is a strong hint that the SYK model answers our question in the affirmative.\footnote{Chaos in the sense of fast thermalization means efficient preparation of Gibbs states via coupling to a bath. See~\cite{chen2021fast} for a quantitative statement connecting the Eigenstate Thermalization Hypothesis and thermalization.}
Unfortunately, it is challenging to sharpen the physics arguments into actual proofs. The only rigorous statement known to us is the recent work of Hastings and O'Donnell~\cite{odonnell_21_optimizing}
, which showed that a low-energy witness with a \textit{constant} ratio approximation of ground energy could be prepared by an efficient quantum algorithm. {An extension of this result to arbitrarily low energies}  would also serve our question. Right now, we do not know any analytic method suitable for the low-temperature regime of the SYK model. 
    \item \textbf{Wigner's Gaussian Unitary Ensemble (GUE)}~\cite{Distribution_roots_58_Wigner}.  If we insist on provable models at low temperatures, we may consider the GUE with dimension $N=2^n$, written in the Pauli string basis:
\begin{align}
    \vH_{GUE} := \sum_{\vsigma \in P } \frac{g_{\vsigma}}{N} \vsigma \quad \text{where} \quad P := \{ \vI, \vsigma^x, \vsigma^y,\vsigma^z\}^{\otimes n} \quad \text{and} \quad g_{\vsigma} \sim \text{i.i.d.~Gaussians}.
\end{align}
This nonlocal and nonsparse Hamiltonian seems unphysical, but it nevertheless served as an mathematical model for heavy nuclei~\cite{nuclear_random}. Together with other random matrix ensembles (see~\cite{anderson_guionnet_zeitouni_2009} for a textbook introduction), the GUE provides a useful model for strongly interacting systems and for quantum information problems thanks to its well-established properties~\cite{RMT_QI_Collins_2016}. 
As the Hamiltonian itself has exponentially many degrees of freedom, by a counting argument, the computational complexity (of low-energy state preparation or Hamiltonian evolution) is exponential $\e^{\Omega(n)}$ with high probability for quantum computers; this is at most polynomially faster than exact diagonalization. 
\end{itemize}

The ensemble we study in this work shares the nice properties of the few-body Pauli and SYK models, as it is sparse and instances can be efficiently specified.  At the same time, like the GUE ensemble, we can accurately approximate the minimal energy and the density of states.   Moreover, we show that polynomial-size quantum circuits are necessary and sufficient to generate low-energy states, up to arbitrarily good approximation ratios of the ground state energy. The main downside of our ensemble is that, like the GUE ensemble, it is nonlocal, so it does not closely resemble the Hamiltonians that readily appear in nature.


\section{Main Results}
\label{sec:main_results}
Now, let us present the model for which we will establish average-case quantum complexity for low-energy states. Consider an independent sum of \textit{a few} random Pauli strings with random sign coefficients:
\begin{align} \label{eqn:Pauli-string-ensemble}
\vH_{PS} &:= \sum_{j=1}^{m}  \frac{r_j}{\sqrt{m}}  \vsigma_j \quad \text{where}\quad \vsigma_j \stackrel{\text{i.i.d.}}{\sim} \{ \vI, \vsigma^x, \vsigma^y,\vsigma^z\}^{\otimes n}\quad \text{and}\quad r_j \stackrel{\text{i.i.d.}}{\sim} \textsc{unif}\{+1,-1\}.
\end{align}
The parameter $m$ will be polynomial in the number of qubits $n$, rather than exponential as in the GUE model.

Our main technical results show that its low-energy states enjoy two-sided bounds on circuit complexity.

\begin{restatable}
[Low-energy states have low complexity]{thm}{main}
\label{thm:low-energy-density}
For any accuracy $\epsilon \ge 2^{-n/c_1}$, let $\vH_{PS}$ be drawn from the Pauli string ensemble~\eqref{eqn:Pauli-string-ensemble} with 
\begin{align}
    m = \left\lfloor  c_2\frac{n^5}{\epsilon^4} \right\rfloor
\end{align}
terms. Then, the following statement holds with probability at least $ 1- \e^{-c_3 n^{1/3}}$ over a random draw $\vH_{PS}$ from the Pauli string ensemble.  We can prepare a low-energy state $\vrho$ such that 
\begin{align}
    \tr[ \vrho \vH_{PS}] \le (1-\epsilon) \cdot \lambda_{\min}(\vH_{PS})\label{eq:low-energy}
\end{align}
using a circuit of size
    $G = \poly(n,\epsilon^{-1}).$
The quantities $c_1, c_2$, and $c_3$ are absolute constants, and $\lambda_{\min}(\vH_{PS})$ denotes the smallest eigenvalue of $\vH_{PS}$ (which is typically negative).
\end{restatable}

See Appendix~\ref{sec:abundance} for the proof of Theorem~\ref{thm:low-energy-density}.  Here, we elaborate on the interesting complexity aspects of this problem:
\begin{itemize}
    
    \item \textbf{Arbitrarily good approximation of the ground energy.} For any polynomially small error $\epsilon\sim \poly(n^{-1})$, there is a polynomially large choice $m = \poly(n)$ for which a state that $\epsilon$-approximates the ground energy can be prepared efficiently at gate complexity $G = \poly(n)$. We hypothesize that the order of quantifiers can be exchanged, which would imply for large enough $m=\poly(n)$ that the low-energy states remain easy for any $\epsilon = \poly(n^{-1})$. For further discussions, see point 2 in Section~\ref{sec:q_advantage}.
    
    \item \textbf{Phase estimation works.} As we will show, the quantum algorithm that produces the low-energy state $\vrho$ is very simple.  Performing phase estimation over the maximally mixed state has a decent chance, at least $\Omega(\epsilon^{3/2})$, of returning a low-energy state obeying~\eqref{eq:low-energy}.  A higher success probability is achieved via repeating the phase estimation step. The Hamiltonian simulation costs at most $\poly(m,\epsilon^{-1})$ gates using off-the-shelf quantum simulation algorithms (e.g, Trotter~\cite{lloyd1996universal} or qDRIFT~\cite{campbell2019random}).
    
    \item \textbf{End-to-end complexity.} This Hamiltonian problem is oracle-free and input-state-free, giving a complete picture. Further, the model description is entirely classical, and an instance can be generated using only $m(2n+1)$-bits of randomness.
    
    \item \textbf{Average-case.} The statement holds \textit{with high probability} over the Hamiltonian ensemble. Indeed, this model can produce an arbitrary local Hamiltonian in the worst case, and we have no control over those instances.
       \item \textbf{Nonlocal, noncommuting Hamiltonians.}
    Most Pauli strings $\{ \vI, \vsigma^x, \vsigma^y,\vsigma^z\}^{\otimes n}$ act nontrivially on $\Theta(n)$ sites, and thus the Hamiltonian is nonlocal. The Hamiltonian is highly noncommutative since random Pauli strings anticommute with each other with probability $\frac{1}{2}$. Intuitively, the Pauli string ensemble is closer to a random matrix than to a local Hamiltonian.
    \item \textbf{Sparse matrices.} From a linear algebra perspective, this model is a sparse, high-rank matrix (which has not been dequantized; see~Section~\ref{sec:q_advantage}). In general, a sparse matrix may not admit a simple Pauli decomposition; nevertheless, the same result extends to signed random $d$-sparse matrices (see~Section~\ref{sec:sparse_matrices}). However, the quantum easiness then requires access to a block encoding~\cite{QSVT_Gily_n_2019} of the Hamiltonian. 

\end{itemize}

 On the other hand, we argue this problem is ``very quantum'' by proving a lower bound on the complexity of preparing low-energy states. As a disclaimer, we do not \textit{prove} classical hardness for state preparation (see Section~\ref{sec:q_advantage}), which is an intriguing open problem that we leave for future work. 
\begin{thm}[small circuit gives bad energy]\label{thm:circui_lower_main}
Fix a circuit architecture with $G$ two-qubit gates (e.g., 1D brickwork layout) with the initial state $\ket{0}$ and consider the family of all reachable states $\mathrm{Circ}(G)$. For any $\epsilon_1 \ge 0$, suppose $ m \le {\epsilon}_1^{2} \cdot 2^{n}$. Then, with high probability over the random draw of the instance $\vH_{PS}$ from the Pauli string ensemble~\eqref{eqn:Pauli-string-ensemble}, 
\begin{align}
G = \tilde{o}(\epsilon_1 \sqrt{m}) \quad \text{implies}\quad \inf_{\ket{\phi} \in \mathrm{Circ}(G)}   \bra{\phi}\vH_{PS}\ket{\phi} \ge \epsilon_1 \cdot  \Expect \lambda_{\min}(\vH_{PS}).
\end{align}
Namely, all possible states $\ket{\psi}\in \mathrm{Circ}(G)$ parameterized by the circuit architecture fail to produce any low-energy state.
The notation $\tilde{o}(\cdot)$ suppresses $\log(m)$-prefactors. 
\end{thm}
See Appendix~\ref{sec:circuit_lower_bounds} for the proof of Theorem~\ref{thm:circui_lower_main}. In other words, we very often need a large circuit $G =\tilde{\Omega}(\epsilon_1 \sqrt{m})$ to describe the low-energy states; they are very entangled and far from product states.\footnote{This statement is analogous to the No-Low energy-Trivial-State conjecture (NLTS)~\cite{freedman_2013}. 
As a disclaimer, we are far from the original context of NLTS, where the Hamiltonians are $\CO(1)$-local, frustration-free, and {topologically ordered}. The proof we provide fails for constant $\CO(1)$-local random Hamiltonian. Indeed, as we mentioned, the 2-body random Pauli model has efficient product state constant-ratio approximation of the operator norm~\cite{Harrow2017extremaleigenvalues}. Recently, the NLTS conjecture for circuit depth $\log(n)$ was proven for certain $\CO(1)$-local Hamiltonian arising from quantum error-correcting codes~\cite{Anshu2022NLTSHF}.} 
Further, our circuit size lower bound uses a direct counting argument, and it suggests the circuit should change over different random instances. Nevertheless, Theorem~\ref{thm:low-energy-density} states the complementary result: an appropriate instance-dependent state can be prepared efficiently using the \textit{simplest} quantum algorithms (Hamiltonian simulation and phase estimation). 



The \textit{main caveat} for our model is that it is nonlocal, unlike most physical Hamiltonians, and our argument is not immediately applicable to local Hamiltonians. Indeed, the spectral properties of the two types of models are different.  As we will show, the Pauli string ensemble has a (compact) semicircular spectrum, while local Hamiltonians tend to have a tail in the spectrum.\footnote{Asymptotically, a Gaussian distribution was known for lattice Hamiltonians~\cite{Brando2015EquivalenceOS} and random $k$-local Pauli Hamiltonians~\eqref{eq:2bodyPauli} for $k= o(\sqrt{n})$ ~\cite{Erd_s_2014}. Nonasymptotically, a spectral tail appears in the 4-local SYK model~\cite{odonnell_21_optimizing}.}
Performing phase estimation with the maximally mixed state would not be able to access the low-energy states far in the low probability tail.  Of course, we hope our results ultimately inspire a better understanding of preparing the low-energy states of local Hamiltonians. For further discussions, see point 3 in Section~\ref{sec:q_advantage}.

Regardless, from a linear algebra and algorithm perspective, random sparse matrices are natural models to study. We emphasize the main goal is to give a transparent toy model showcasing what quantum computers are good at, especially given recent developments in dequantization.

\section{Proof ideas}
\label{sec:proof_ideas}
\begin{figure}[t]
    \centering
    \includegraphics[width=0.6\textwidth]{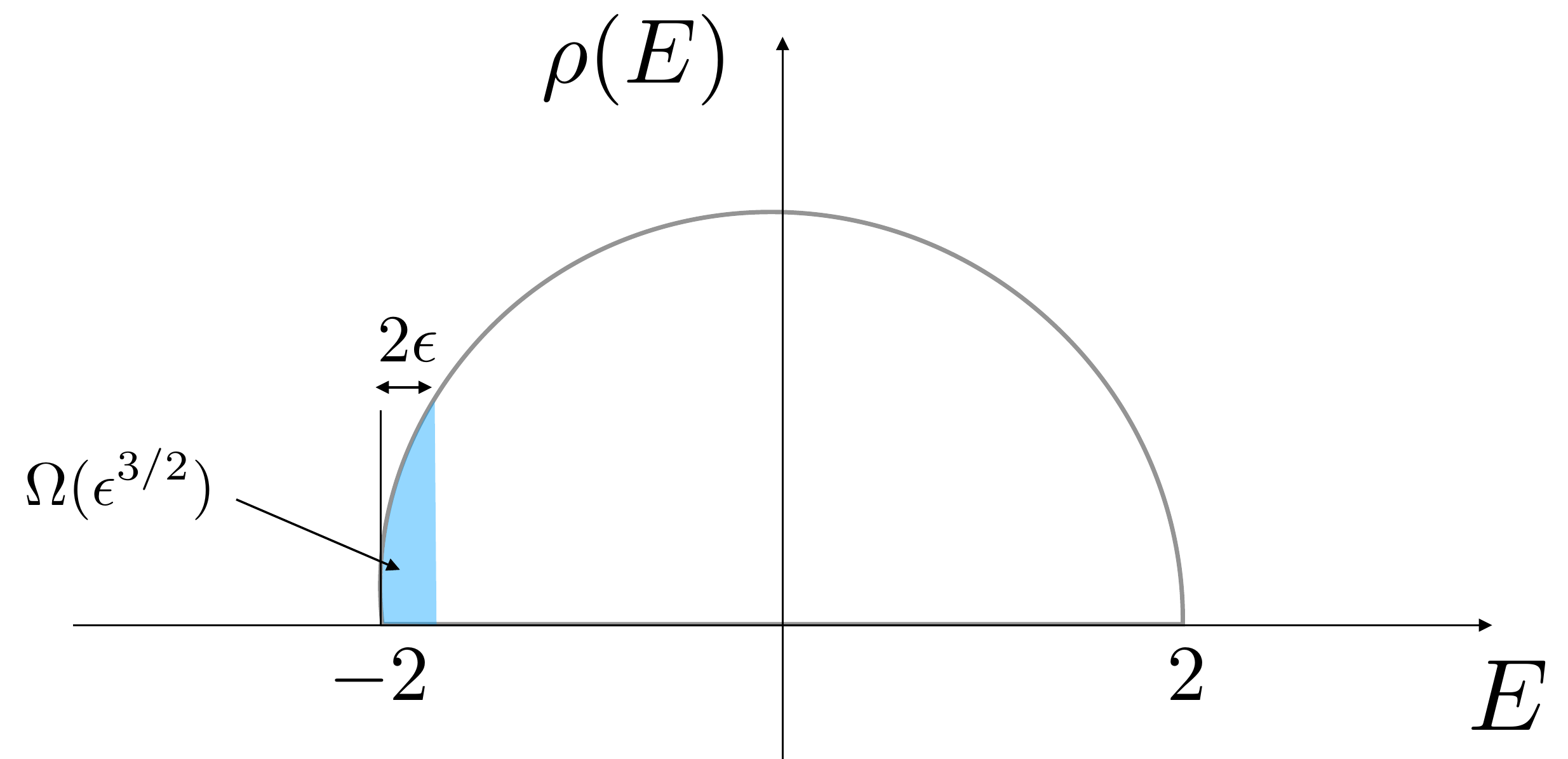}
    \caption{\textbf{Abundance of low energy states.} The contour illustrates the density $\rho$ versus the energy level $E$ for a semicircular distribution, which is (in the large dimension limit) the spectral distribution for the GUE. The semi-circle spectral density implies the abundance of states near the ground energy. Performing phase estimation over the maximally mixed state gives a state with low energy $-2(1-\epsilon)$ with a decent probability $\Omega(\epsilon^{3/2})$.
    }
    \label{fig:semicircle_integral}
\end{figure}

Given a general strongly interacting Hamiltonian, it seems daunting to control its behavior.  However, we can make an exception in the case of certain random matrix ensembles where the matrices have predictable spectral properties.  For example, it is well known that GUE matrices (see Appendix~\ref{sec:GUE}) have a definite maximal eigenvalue and a semicircular spectral density $\rho(x)$:
\begin{align}
\norm{\vH_{GUE}} \approx 2\quad \text{and}\quad \rho(x) \approx \frac{\sqrt{4-x^2}}{2\pi} \quad \text{(up to negligible fluctuation and with high probability)}.
\end{align}
See Figure~\ref{fig:semicircle_integral}.
Indeed, this fact alone hints that the low-energy states have a \textit{nonnegligible} density, \textit{independent} of the system size:
\begin{align}
    \int_{-2}^{-2+2\epsilon} \rho(x) \,\mathrm{d} x =\Theta( \epsilon^{3/2} ).
\end{align}
In terms of complexity, directly running phase estimation on the maximally mixed state returns a low-energy state with decent probability: $\Omega(\epsilon^{3/2})$.
The core of our argument is that the spectrum of the Pauli string ensemble~\eqref{eqn:Pauli-string-ensemble} looks a lot like the
spectrum of a GUE matrix (Figure~\ref{fig:good_bad}). As a consequence, it is also ``easy'' to find the low-energy states of the Pauli string ensemble.


How can we prove that the Pauli string ensemble also has a semicircular spectral density?
The entire argument then boils down to a \textit{universality principle}:
\begin{quote}
\textit{The Pauli string ensemble, at moderately large $m$, mimics ``smooth'' properties of the GUE ensemble, including the maximum eigenvalue and the coarse-grained spectral density}.
\end{quote}
%
The mathematical argument is based on techniques from nonasymptotic random matrix theory (Section \ref{sec:random-matrix}).  Several novel results are required to address some of the particular challenges that arise in the quantum information problem.

\begin{figure}[t]
    \centering
    \includegraphics[width=0.9\textwidth]{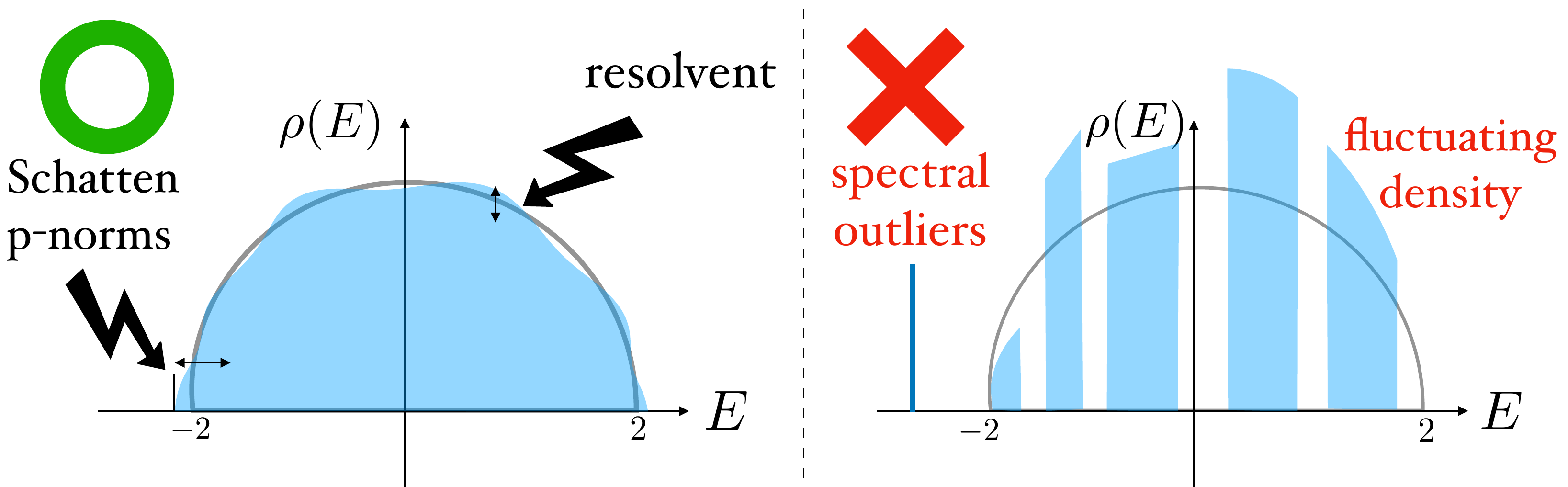}
    \caption{\textbf{(Right) What could have gone wrong.} For a generic matrix, the phase estimation strategy would not find a low-energy state if there are spectral outliers or if the spectral density gets too small near the ground energy. 
    \textbf{(Left) Almost a semi-circle.} For the Puali string ensemble, we control the ground energy by Schatten $p$-norms and control the spectral density by the resolvent. Both values are comparable to the GUE, which has a favorable semicircular spectrum. Therefore, enough states remain near the ground energy, and thus phase estimation efficiently finds them.
    }
    \label{fig:good_bad}
\end{figure}

 Our first main result states that the trace polynomial moments of the Pauli string ensemble almost coincide with the corresponding moments of the GUE.  In particular, by choosing a large enough moment, we can also compare the spectral norms of the two matrices. Throughout this work, we consider the normalized $p$-norms
\begin{equation} 
\norm{\vO}_p:= \left(\btr\labs{\vO}^p\right)^{1/p} \quad \text{and} \quad 
    \vertiii{\vO}_p := \left(\Expect  \btr\labs{\vO}^p \right)^{1/p} \quad \text{where}\quad \btr := \frac{1}{\tr[\vI]}\tr .\label{eq:define_pnorms}
\end{equation}
We often denote $N:=\tr[\vI]=2^n$. We have the understanding that $\vertiii{\vO}_{\infty} \coloneqq \mathrm{ess\,sup}\, \norm{ \vO }$.




    


\begin{restatable}[$p$-norms and operator norm]{thm}{pnormsPauli}
\label{thm:pnorms_pauli}
Let $p \in 2\mathbb{N}$ be an even natural number.
The random Pauli string ensemble~\eqref{eqn:Pauli-string-ensemble} satisfies the norm bound
\begin{align}
    \labs{ \vertiii{ \vH_{PS} }_{p} - \vertiii{\vH_{GUE} }_{p}}\lesssim \left(\frac{p^{3/4}}{m^{1/4}} + \frac{p}{ \sqrt{m}} \right)\left( 1+ \frac{p^{3/4}}{2^{n/2}} \right)\label{eqn:PSE-moment-intro}.
\end{align}
The symbol $\lesssim$ suppresses constant factors. Furthermore, for $0 \le \epsilon \le 1/2$ and $m \le 2^{2n}$, there exist constants $c_1,c_2 > 0$ where
\begin{align}
        m \ge c_1 \frac{n^3}{\epsilon^4}
\quad \text{ensures}\quad \Pr\L( \norm{\vH_{PS}} \ge 2(1+\epsilon) \R) \le \exp( -c_2 n ). 
\end{align}
\end{restatable}

%
%
See Appendix~\ref{sec:pauli_moments} for the proof of Theorem~\ref{thm:pnorms_pauli}.
For a fixed moment $p$ that may depend on the number $n$ of sites, the right-hand side of~\eqref{eqn:PSE-moment-intro} decays with the number $m$ of terms in the Hamiltonian.  Applying Markov's inequality for $p = \Omega(\log(N)) = \Omega(n)$ and choosing $m = \poly(n)$, we obtain a tail bound for the spectral norm.

Comparing the spectral densities of the two ensembles requires a more difficult argument.
Ideally, we are interested in projectors to Hamiltonian eigenstates $ \ket{\phi}\bra{\phi}\label{eq:ideal_projector}$. However, exact eigenstate projectors are tricky to handle. Instead, we consider the \textit{resolvent}, which probes the ``coarse-grained'' energy projector at energies $\omega \pm \CO(\eta)$.  We define
\begin{align} \label{eqn:resolvent-intro}
    \vR_{\omega,\eta}(\vH) : = \frac{1}{\vH - \omega +\ri \eta} = \ri\int_0^{\infty} \e^{\ri(\vH -\omega )t - \eta t} dt \quad \text{as a proxy for}~ \quad \frac{1}{\eta}\sum_{\ket{\phi}} \indicator\{ \labs{E(\phi)-\omega} \le \eta\} \cdot \ket{\phi}\bra{\phi},
\end{align}
where $\indicator$ is the indicator function. We often suppress parameter dependencies by writing $\vR \coloneqq \vR_{\omega,\eta}(\vH)$. For intuition, the resolvent is diagonal in the Hamiltonian basis and spikes at energy $\omega$ with width $\CO(\eta)$. See Figure~\ref{fig:semi_circle_resolvent}. 

 

However, if we are especially interested in the states near certain energy $\omega$, the resolvent is not localized enough because the filter $E \mapsto 1/\labs{E-\omega}$  decays too slowly as a function of the energy.\footnote{This is closely related to the phase estimation amplitude profile where the width $\eta$ is roughly the resolution (inversely proportional to the runtime).}
Instead, we can take the trace of resolvent powers so that the tail decays at the faster rate $\sim \labs{E-\omega}^{-p}$.  That is,
\begin{align}
    \frac{\eta^{p}}{N} \tr \labs{\vR}^p = \frac{1}{N} \sum_{E} \frac{\eta^p}{\labs{E-\omega+\ri \eta}^{p}} \quad \text{as a proxy for}\quad \rho(\omega) \cdot \frac{\eta}{\sqrt{p}} 
\end{align}
where the energy $\frac{\eta}{\sqrt{p}}$ is roughly the window where the weight $\labs{\vR}^p\eta^{p}$ remains large $\Omega(1)$. 

\begin{thm}[Comparing the resolvent moments]\label{thm:resolvent_paulis}
Let $p \in 2\mathbb{N}$ be an even natural number.  The resolvent~\eqref{eqn:resolvent-intro} of the random Pauli string ensemble~\eqref{eqn:Pauli-string-ensemble}, written $\vR_{PS}$, compares with the resolvent $\vR_{GUE}$ of the GUE:
\begin{align} \label{eqn:resolvent-pauli-intro}
        \labs{\vertiii{\vR_{PS}}_{p}- \vertiiismall{\vR_{GUE}}_{p}}
        \lesssim \left(\frac{p^4}{\eta^5m^2} + \frac{p^3}{\eta^5m}\right)\left( 1+ \frac{p^3}{2^{2n}}\right).
\end{align}
The symbol $\lesssim$ suppresses absolute constants.
\end{thm}
See Appendix~\ref{sec:resolvent} for the proof of Theorem~\ref{thm:resolvent_paulis} in a more general setting.
For moderately large $m$ (depending on the distance $\eta$ from the real line and the power $p$), the formula~\eqref{eqn:resolvent-pauli-intro} controls the expected {spectral density}, filtered by the resolvent:
\begin{align}
    \Expect  \tr \labs{\vR_{PS}}^p \approx \Expect  \tr \labs{\vR_{GUE}}^p.
\end{align}
%
%
Since we want to make a statement that holds with high probability over realizations of the Pauli string ensemble, we also need to prove that the quantity $\tr \labs{\vR}^p$ concentrates near its expectation $\Expect  \tr \labs{\vR}^p$ (i.e., the spectral density does not fluctuate too much); see Theorem~\ref{thm:resolvent_con_good_lambda}.
Lastly, since individual resolvents probe the local density, we may probe the \textit{integrated} spectral density by placing consecutive resolvents. The abundance of low-energy states then implies phase estimation succeeds with a decent chance. 

\begin{figure}[t]
    \centering
    \includegraphics[width=0.6\textwidth]{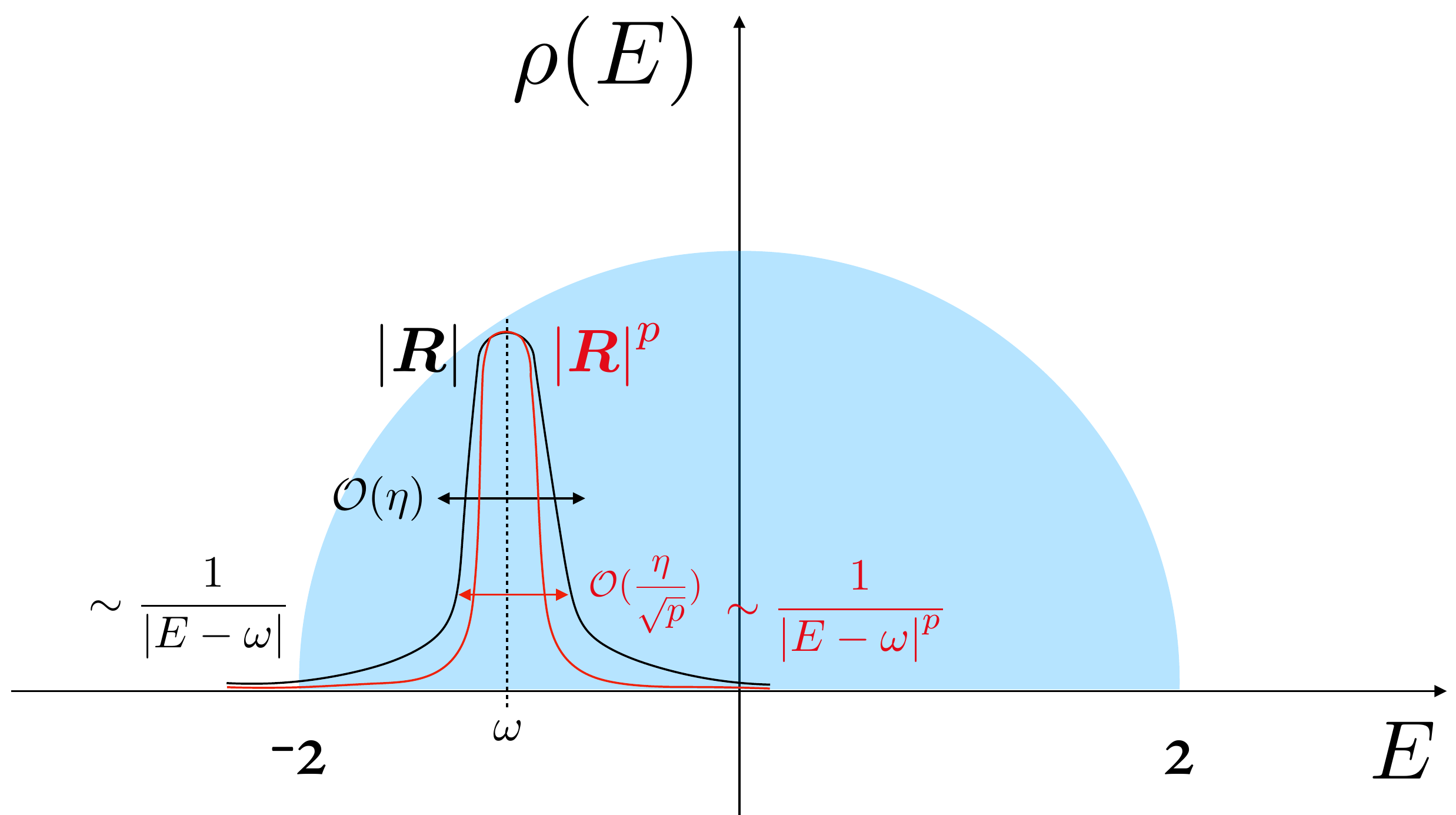}
    \caption{\textbf{Probing the spectrum by resolvents.}  
    The resolvent (black curve) centered at energy $\omega$ with resolution parameter $\eta$ filters out energies distant from $\omega$. Taking powers of the resolvent (red curve) focuses the filter on a narrower region around $\omega$.
    }
    \label{fig:semi_circle_resolvent}
\end{figure}

\section{New results in nonasymptotic random matrix theory}\label{sec:random-matrix}

Our results for the Pauli string ensemble fall into
the category of nonasymptotic universality laws
for random matrices.  This section provides some
context for these results, as well as some details
about the argument.



Asymptotic universality laws are among the celebrated
classical achievements of random matrix theory (RMT).
For example, Wigner showed that the semicircle law
is the limiting spectral distribution of a
(standardized) symmetric matrix with i.i.d.~Rademacher entries above the diagonal.
The universality law for the Wigner matrix states that
the detailed distribution of the entries does not affect
the limiting spectral distribution, provided the first four
moments are bounded.
Subsequently, researchers obtained \emph{nonasymptotic}
comparisons between the spectrum of a Wigner-type matrix and the semicircle distribution.
For surveys, see the monographs~\cite{BS10:Spectral-Analysis,PS10:Eigenvalue-Distribution}.


Our approach depends on a nonasymptotic comparison between
the spectrum of the Pauli string ensemble~\eqref{eqn:Pauli-string-ensemble} and a GUE matrix, whose spectral distribution approximately follows a semicircle law.  This type of result does not fall within the scope of classical universality laws because the Pauli string ensemble barely has any randomness, let alone independent entries.  To implement our program, we first observe that the low-order moments of Pauli string ensemble match the low-order moments of a GUE matrix:
\[
\begin{aligned}
\Expect[ \vH_{PS} ] &= \Expect[ \vH_{GUE} ]; \\
\Expect[ \vH_{PS} \otimes \vH_{PS} ] &= \Expect[ \vH_{GUE} \otimes \vH_{GUE} ]; \\
\Expect[ \vH_{PS} \otimes \vH_{PS} \otimes \vH_{PS} ] &= \Expect[ \vH_{GUE} \otimes \vH_{GUE} \otimes \vH_{GUE} ].
\end{aligned}
\]
For a smooth statistic $f$ of the random matrices, we can take advantage of this coincidence by means of Lindeberg's exchange principle.  Each of the random matrix models can be expressed as a sum of i.i.d.~random matrices, and we can interpolate between the two models by swapping one summand at a time.  At each step, we can control the change between the two models by expanding $f$ as a Taylor series to expose the polynomial moments.  The terms in these expansions cancel through the third order, leaving a fourth-order error.  Our argument is quite different from recent applications~\cite{Cha05:Simple-Invariance,KM11:Applications-Lindeberg} of the Lindeberg
principle in RMT.

In more detail, we consider two random Hermitian matrices $\vH$ and $\tilde{\vH}$ that can be written as sums of independent, centered random matrices (all of the same dimension):
\begin{equation} \label{eqn:rdm-sum-intro}
\vH = \sum_{i=1}^m \vA_i
\quad\text{and}\quad
\tilde{\vH} = \sum_{i=1}^m \tilde{\vA}_i.
\end{equation}
Although less familiar than the classical random matrix ensembles, the independent
sum model is much more flexible and has a wide scope of applicability;
see~\cite{tropp2015introduction} for examples.
Suppose that the low-order polynomial moments of the summands match.  That is, 
\begin{equation} \label{eqn:moment-match-intro}
\Expect \vA_i = \bm{0}
\quad\text{and}\quad
\Expect \vA_i^{\otimes k} = \Expect \tilde{\vA}_i^{\otimes k}
\quad\text{for $k = 1, \dots, t$ and $i = 1, \dots, m$.}
\end{equation}
For example, the first three moments of a random Pauli string match the first three moments of a GUE matrix.  More generally, constructive models in quantum information theory can match an arbitrary number $t$ of moments, similar to the case of a unitary $t$-design.  Our work shows how to compare the spectral properties of models with many matching moments.

Our first universality result compares the trace polynomial moments of the two random matrices.  These results allow us to control the spectral norm of the random matrices.

\begin{restatable}[Universality for moments]{thm}{universalitymoments}
\label{thm:universality_moments}
Consider two families $(\vA_i)$ and $(\tilde{\vA}_i)$ of independent random, Hermitian matrices whose moments match~\eqref{eqn:moment-match-intro} up to order $t \geq 2$, and introduce the sums $\vH$ and $\tvH$ as in~\eqref{eqn:rdm-sum-intro}.  Define the statistics
\begin{align*}
L_{p,k} &\coloneqq \left[ \sum_{i=1}^m \L( \vertiii{\vA_i}_p^k +  \vertiii{\smash{\tvA_i}}_p^k \R) \right]^{1/k}
\quad\text{for $k \geq 1$ and $p \geq 2$;} \\
\sigma^2 &\coloneqq \sum_{i=1}^m \big(\norm{ \Expect \vA_i^2 } + \norm{ \smash{\Expect \tvA_i^2} } \big).
\end{align*}
Then, for each even natural number $p \in 2\mathbb{N}$, we have the bounds
\begin{align} \label{eqn:anthony-moments}
    \labs{ \vertiii{ \vH }_{p} - \vertiii{ \smash{\tvH} }_{p}} 
    \quad\leq\quad \begin{dcases}
        2 p^{t/(t+1)} \cdot L_{p,t+1} + 2p \cdot L_{p, p}; \\
        2 p^{t/(t+1)} \cdot \big(\sigma^2 L_{\infty,\infty}^{t-1}\big)^{1/(t+1)} + 2p \cdot L_{p,p}.
    \end{dcases}
\end{align}
The $p$-norm is defined in~\eqref{eq:define_pnorms}.
\end{restatable}


The proof of Theorem~\ref{thm:universality_moments} appears in Section~\ref{sec:moments}.  Theorem~\ref{thm:pnorms_pauli}
follows when we instantiate this result for the Pauli string ensemble~\eqref{eqn:Pauli-string-ensemble} and the GUE.  

We can obtain simpler versions of this result if we pass to the
uniform bound $L_{p, \infty} = \max_i \{ \vertiii{\vA_i}_p, \vertiii{\smash{\tvA_i}}_p \}$
on the summands.  For example,
\begin{align}
\labs{ \vertiii{ \vH }_{p} - \vertiii{ \smash{\tvH} }_{p}} 
    \quad\lesssim\quad \big( 1 + (m/p) \big)^{1/(t+1)} \cdot p L_{p,\infty}. 
\end{align}
Here, the symbol $\lesssim$ suppresses absolute constants only.
Heuristically, we should think about $p \ll m$, 
so there are reductions in the error from matching more moments (i.e., increasing $t$).

Our second universality result provides a comparison for powers of the resolvents of independent sums.  Define
\begin{equation} \label{eqn:rdm-sum-resolvent-intro}
\vR \coloneqq (\vH - \omega + \iunit \eta)^{-1}
\quad\text{and}\quad
\tilde{\vR} \coloneqq (\tilde{\vH} - \omega + \iunit \eta)^{-1}
\quad\text{where $\omega \in \mathbb{R}$ and $\eta > 0$.}
\end{equation}
The random matrices $\vH$ and $\tilde{\vH}$ are defined in~\eqref{eqn:rdm-sum-intro}.


\begin{restatable}[Universality for resolvent moments]{thm}{universalityResolvent}\label{thm:universality_resolvent} 
Instate the assumptions and notation of Theorem~\ref{thm:universality_moments}.  For each even natural number $p \in 2 \mathbb{N}$, the polynomial moments of the resolvent~\eqref{eqn:rdm-sum-resolvent-intro} are related by 
\begin{align}
        \labs{\vertiii{\vR}_{p}- \vertiiismall{\tvR}_{p}} \quad\lesssim\quad  \frac{1+ (m/p)}{\eta} \cdot \left( \frac{pL_{3p,\infty}}{\eta} \right)^{t+1}.
\end{align}
The symbol $\lesssim$ suppresses constants depending only on $t$.
\end{restatable}

See Appendix~\ref{sec:resolvent} for the proof of Theorem~\ref{thm:universality_resolvent}.
We obtain Theorem~\ref{thm:resolvent_paulis}
by instantiating the result for the Pauli
string ensemble and the GUE.

The resolvent moment comparison (Theorem~\ref{thm:universality_resolvent}) is not sufficient
to guarantee that a random realization $\vH_{PS}$
of the Pauli string ensemble places significant
density on the low-energy states.  To achieve this
goal, we must also show that $\vertiii{ \vR_{PS} }_p$ concentrates near its expected value.  This claim requires a separate
argument (Theorem~\ref{thm:resolvent_con_good_lambda}).
The results on concentration of the trace moments of
the resolvent are new.


\subsection{Related work}

The field of RMT has historically
focused on asymptotic limit laws for the spectral
density of matrices from the classical ensembles
(Wigner, Wishart, Jacobi, etc.).  In this setting,
there has also been a significant amount of research
on rates of convergence, and some of these results
can be interpreted as nonasymptotic universality
laws.  For example, see Bai \& Silverstein~\cite[Chap.~8]{BS10:Spectral-Analysis}.

In the last few years, researchers have recognized
that the scope of the universality phenomenon extends
well beyond the classical matrix ensembles.  In particular,
we have started to develop a deeper understanding of the
independent sum model.  Tropp obtained the first general
result of this type~\cite{Tro18:Second-Order-Matrix}.
His theory covers a sum of independent Gaussian random
matrices, and it provides conditions under which
the polynomial moments approximate the moments of
the semicircle distribution.  Building on Tropp's work,
Bandeira et al.~\cite{Bandeira2021MatrixCI} developed a method for
comparing a sum of independent Gaussian random matrices
with a free probability model, which can capture a
wider range of spectral distributions.  With some effort,
the techniques from these two papers can likely be applied
to the Gaussian variant of the Pauli string
ensemble~\eqref{eqn:Pauli-string-ensemble}
to obtain results similar to our main theorems.

The most immediate precedent for our work is a recent
preprint by Brailovskaya \& van Handel~\cite{Tatiana_22_universality}.
Their paper compares an independent sum $\vH$ of random matrices
with an independent sum $\vG$ of Gaussian random matrices, where
corresponding summands share the same mean and covariance:
\begin{equation} \label{eqn:gauss-compare-intro}
\vH = \sum_{i=1}^m \vA_i
\quad\text{and}\quad
\vG = \sum_{i=1}^m \tilde{\vA}_i
\quad\text{where}\quad
\Expect[ \vA_i ] = \Expect[ \tilde{\vA}_i ]\quad\text{and}\quad
\Expect[ \vA_i \otimes \vA_i ]= \Expect[ \tilde{\vA}_i \otimes \tilde{\vA}_i].
\end{equation}
The main result of the paper~\cite{Tatiana_22_universality}
provides conditions to guarantee that the two random matrix
models have similar polynomial moments and polynomial
resolvent moments.

\begin{thm}[Universality of moments and resolvent moments~\cite{Tatiana_22_universality}] \label{thm:tatiana}
Consider two random matrix models as in~\eqref{eqn:gauss-compare-intro}.  Define the statistics
\[
v :=  \lnorm{\sum_{i=1}^m\Expect \vA^2_i}
\quad\text{and}\quad L_{\infty} := \max\nolimits_i \norm{\vA_i}.
\]
Then, for every even natural number $p \in 2\mathbb{N}$, the polynomial moments and resolvents satisfy the bounds
\begin{align}
    \labs{ \vertiii{ \vH }_{p} - \vertiii{ \vG }_{p}} &\quad\lesssim\quad  \left(p^2 v L_{\infty}\right)^{1/3} +   pL_{\infty}; \label{eqn:tatiana-moments} \\
    \labs{ \vertiii{ \vR_{\vH} }_{p}- \vertiii{ \vR_{\vG} }_{p}}
        &\quad\lesssim\quad \frac{ p^2 v L_{\infty} + p^3 L_{\infty}^3 }{ \eta^4 }.
\end{align}
The $p$-norm is defined in~\eqref{eq:define_pnorms},
and the symbol $\lesssim$ suppresses absolute constants.
\end{thm}


The proof of this result uses a version of Stein's method, inspired by~\cite{LP09:Central-Limit}.
The basic technique is to interpolate smoothly between the two random matrix models,
preserving the second moments along the interpolation path.
To control the derivative of a spectral function along the path,
the authors use a cumulant expansion along with bounds on the
higher derivatives of the function.

It is fruitful to compare the bounds~\eqref{eqn:anthony-moments} and \eqref{eqn:tatiana-moments}.
The variance parameter $v$ in Theorem~\ref{thm:tatiana} is never larger than
the variance parameter $\sigma^2$ in Theorem~\ref{thm:universality_moments}
because the norm is inside the sum in $v$.  The two quantities $\sigma^2$
and $v$ coincide for i.i.d.~sums, but they can differ by a factor as large
as the ambient dimension $N$ in general.  The differences between the tail
parameters ($L_{p,p}$ and $L_{p,\infty}$ and $L_{\infty}$) are not an essential
feature of the analysis; we have stated the simplest versions of the results,
rather than the optimal versions.


On the other hand, the approach in Theorem~\ref{thm:tatiana} cannot provide
more refined comparisons for random matrix models that match beyond the second
moment (except perhaps when the third moments are identically zero).  There
are intrinsic reasons that continuous interpolation does not seem to
extend beyond second moments (Appendix~\ref{sec:difficult_higher}).
In contrast, the method based on Lindeberg exchange gracefully handles matching
moments of any order.

As we will see (Section~\ref{sec:examples}), there are some
natural settings where higher-order moments coincide.  The resulting
higher-order error bounds improve over the second-order bounds.
In the setting of quantum information, we often need to take the moment parameter 
$p \sim \log N \sim n$, so this improvement is significant.

In addition, our argument is conceptually and technically simpler than the approach
based on Stein's method and cumulant expansions.  As a consequence,
it may be easier to extend to other settings,
and it may have a different scope of application.
Altogether, our work contributes to the emerging toolkit for
nonasymptotic RMT.

\subsection{Further examples} \label{sec:examples}


Our universality results apply to many different families
of random matrix models, including examples that may not
resemble the Gaussian models that are central
to the comparison in~\cite{Tatiana_22_universality}. For quantum computing applications, these families could potentially capture realistic sparse matrices better than random Pauli string sums. However, they generally require access to an additional block-encoding, which we do not discuss in this work.

\subsubsection{Comparing sparse matrices with GUE}\label{sec:sparse_matrices}

In addition to the random Pauli string ensemble, we can describe another family of sparse random matrices that also matches the low moments of GUE. Therefore, the universality results (Theorem~\ref{thm:universality_moments}, Theorem~\ref{thm:universality_resolvent}) show
that these models nearly follow a semicircular distribution.

\begin{defn}[Permutations with complex signs]
A random complex signed permutation matrix is the product of a uniformly random permutation matrix $\vP$ and a diagonal matrix $\vD$ with complex signs:
\begin{align}
    \vQ:= \vD \vP  \quad \text{where} \quad  \vD_{ab} = \delta_{ab} \frac{r_a + \ri r'_a}{\sqrt{2}}\quad \text{and} \quad r_a ,r'_a \stackrel{i.i.d.}{\sim} \{1,-1\}.
\end{align}
\end{defn}


\begin{prop}[Complex signed permutations] \label{fact:complex_sign_match_GUE} Consider random matrices $\vA$ and $\tvA$
that take the form
\begin{align}
        \vA &:= \frac{\vQ + \vQ^{\dagger}}{\sqrt{2}} \quad \text{where $\vQ$ is a complex signed permutation} \\
        \tvA &\sim \vH_{GUE}.
\end{align}
For these models, the first three moments match:
\begin{align}
    \Expect [\vA^{\otimes k}] = \Expect [\tvA^{\otimes k}]\quad  \text{for each}\quad k = 1,2,3.
\end{align}
\end{prop}

See Appendix~\ref{sec:proof_complex_match} for the calculation. One may also consider random real signed permutations,\footnote{ The random signed permutation is defined by $\vQ':= \vD' \vP$ where $\vD'_{ab} = \delta_{ab} r_a$ and $r_a \stackrel{i.i.d.}{\sim} \{1,-1\}$.}
which match the first three moments of the Gaussian Orthogonal Ensemble (GOE).



\subsubsection{Higher moments matching}




Even though higher-moment matching examples are less common in the wild, we can describe several pairs of models that match up to arbitrarily high moments. The first example considers conjugating a fixed matrix by random unitaries:
\begin{align}
        \tvA_i&:= \frac{1}{\sqrt{m}} \tvU_{i}\vsigma\tvU_{i}^{\dagger}  \quad \text{where} \quad \tvU_{i} \sim \text{haar};\\
        \vA_i&:= \frac{1}{\sqrt{m}}\vU_{i}\vsigma\vU_{i}^{\dagger}  \quad \text{where} \quad \vU_{i} \sim \text{unitary $t$-design.}
\end{align}
Indeed, if we take the unitaries $\vU_i$ to be the Clifford circuits (exact 3-design) and $\vsigma$ to be a fixed Pauli string, we nearly obtain the Pauli string ensemble (up to the identity element that cannot be produced by conjugation). However, beyond Clifford circuits, we do not know other examples where the matrices $\vA_i$ remain sparse. 

If we insist on sparse matrices, here is another example. 
\begin{align}
        \tvA_i&:= \frac{\tvQ_i + \tvQ^{\dagger}_i}{\sqrt{2m}} \quad \text{where}\quad \tvQ_i \sim \quad \text{i.i.d.~complex signed permutations;}\\
        \vA_i&:= \frac{\vQ_i + \vQ^{\dagger}_i}{\sqrt{2m}} \quad \text{where}\quad \vQ_i \sim \quad \text{i.i.d.~$t$-wise independent complex signed permutations.}
\end{align}
In this context, $t$-wise independent permutation is exactly the $t$-th moment matching condition
$\Expect \vQ_i^{\otimes t} = \Expect \tvQ_i^{\otimes t}$. 
Exact and \textit{approximate} constructions for both $t$-designs~\cite{Brand_o_2016,exactdesign_Nakata_2021} and $t$-wise independent permutations~\cite{almost_kwise_13_Alon} are available in the literature.
We leave for future work for a careful analysis of the approximate case where very few random bits are needed.

\section{Comments on dequantization and quantum advantage} \label{sec:q_advantage}


In this section, we comment on the \textit{classical} complexity for the low-energy problem. The flavor differs from local Hamiltonian problems because our model is highly nonlocal and has a semi-circular spectrum. 

\textbf{1. How far does dequantization go?}
As we mentioned, recent developments in dequantization show that many linear algebra tasks can be efficiently solved assuming certain classical access to a quantum state. In particular, existing results consider low-rank matrices for various tasks~\cite{Tang2018QuantumPC, Chia2018QuantuminspiredSC, Chia2019SamplingbasedSL,quantum_inspired_Tang_2019} or high-rank matrices but with constant accuracy~\cite{Gharibian_2022}. 

In the setting of Theorem~\ref{thm:low-energy-density}, we provide an efficient classical witness for the optimum if the accuracy $\epsilon >0$ is an arbitrarily fixed \textit{constant} (with polynomially large $m$). The idea is a simple polynomial approximation. However, the cost of manipulating the witness is $\exp(\tilde{\Omega}(1/\sqrt{\epsilon}))$, which scales poorly with the constant $\epsilon$.
\begin{prop}[Efficient classical witness at arbitrary constant accuracy]\label{prop:efficient_classical_witness}
For any $\epsilon$ and large enough $m = \Omega(\poly(n,\epsilon^{-1}))$, there is a degree $d=\CO(1/\sqrt{\epsilon})$ polynomial $p_d(x)$
such that the associated ansatz state has low energy
\begin{align}
    \vrho \propto p_d(\vH)^2 \quad \text{such that} \quad \tr[\vrho\vH] \le (1-\epsilon) \lambda_{\min}(\vH).
\end{align}
Further, this can be efficiently verified classically in runtime 
\begin{align}
    \CO((dm)^{2d} n d).
\end{align}
\end{prop}
\begin{proof}[Proof of Proposition~\ref{prop:efficient_classical_witness}]
Using power series approximation (i.e., Taylor expansion) for the Gibbs state gives a suboptimal degree $d = \CO(1/\epsilon)$. A better degree $d = \CO(1/\sqrt{\epsilon})$ can be achieved using Chebychev's polynomial approximation of the Gibbs state. The verification algorithm simply evaluates all the $(dm)^{2d}$ terms in the ansatz state $\vrho$. Each of the $(dm)^{2d}$ terms require $2d$-multiplications of Pauli strings, each with cost $\CO(n)$. 
\end{proof}

This indicates that eigenstates ``far from the ground state'' have polynomial classical complexity; this is reminiscent of the cost of dequantization methods~\cite{Gharibian_2022} in the context of ground energy estimation given good trial states. Still, the above classical polynomial witness gets stuck at a constant approximation ratio, while the quantum algorithm has no problem going to better and better accuracy\footnote{More carefully, our current results requires the number of terms $m$ to grow with the desired accuracy.}.

\textbf{2. Give me a decision problem!}
To really talk about quantum advantage, ideally one wants a problem with classical inputs and outputs, 
 especially a decision problem. A candidate problem is to compute an approximation to the ground energy of our model.
However, since our problem has randomness, we expect the spectrum to be concentrated around the semicircle. {If the spectrum were exactly the semicircle, a classical algorithm could simply output the deterministic value.}  Therefore, the classical hardness, if it exists, must originate from the instance-to-instance \textit{fluctuation} of the spectrum away from the semicircle density, and that is why we need the accuracy $\epsilon=1/\poly(n)$ to be small while the number of terms $m = \poly(n)$ is not too large (otherwise the fluctuation becomes too small and predictable). 

Acknowledging the above, a candidate problem for quantum advantage is deciding the density of states to high precision. It also converts to a binary decision problem by setting a threshold.

\begin{quest}[Task: Deciding the density of states]
Given a Hamiltonian sampled from the Pauli string ensemble {and a small parameter $\epsilon$}, output the number of states at a small energy interval 
\begin{align}
 [-\delta, \delta] \subset [-2, 2] \quad \text{ up to multiplicative error} \quad \epsilon.
\end{align}
Is it classically hard for some $m = \poly(n)$, $\epsilon = \poly(n)^{-1}$, $\delta = \poly(n)^{-1}$?
\end{quest}

There is a quantum algorithm that succeeds with gate complexity $\poly(\epsilon^{-1},\delta^{-1},m)$: our concentration argument for low-energy density of states (\eqref{eq:WTS} in the proof of Theorem~\ref{thm:low-energy-density}) also implies that for each $\delta$, there is a polynomially large $m = \poly(\delta^{-1},n)$ such that the local density $[-\delta,\delta]$ is at least half of that of the semicircle. Therefore, phase estimation samples from this interval with $\Omega( \delta^{-1})$ success probability. Repeated trials\footnote{ Taking $O( \delta^{-1}\epsilon^{-2})$ samples, one can estimate number of states to multiplicative error $\epsilon$. 
} 
give a high-confidence estimate of the density of states to error $\epsilon$ with $\poly(1/\epsilon)$ algorithmic cost.
Why consider the problem of approximating the density of states and not approximating the ground state energy? Right now, we do not have control over the spectrum very close to the extreme eigenvalues; for fixed $m,n$, our current results do not rule out the possibility of a small spectral gap $\Omega(m^{-1/4}n^{5/4})$;
while we believe the spectral gap is exponentially small, the proof will require further developments in nonasymptotic random matrix theory.

Proving classical hardness for the density of states problem, e.g.~by reduction from a problem already known to be hard, is more elusive. A general proof might be too much to ask for as it would give a computational quantum advantage for an oracle-free average-case decision problem, something for which no other examples are known. Still, it would be interesting to provide arguments for it. A concrete step is proving that the spectrum has a large enough instance-to-instance fluctuation away from the semicircle distribution such that the classical algorithm cannot succeed simply by always outputting the average value. We believe this to be true and it would be interesting to test it numerically. However, a proof of it would require further developments in nonasymptotic random matrix theory. 

\textbf{3. Quantum chaos and quantum advantage.}
Our work fits into the broader question of whether quantum chaos could be a source of quantumly easy problems and perhaps a quantum computational advantage. As we mentioned, Hastings and O’Donnell~\cite{odonnell_21_optimizing} made concrete progress on the SYK model, a prominent toy model of quantum chaos, by providing a low-energy witness where Gaussian states are known to fail. Their results would serve {the question at hand} even better if the classical hardness argument can be improved or if the Hamiltonian remains provably easy near the ground state. The latter seems plausible on physical grounds as it remains ``chaotic'' near the ground energy. Indeed, if one were to formally assume quantum chaos in terms of the Eigenstate Thermalization Hypothesis (ETH), one may prove that {preparing low-energy states is quantumly easy because Gibbs sampling at low temperatures is efficient on a quantum computer~\cite{chen2021fast}.}

Our work made progress in capturing quantum chaos and its consequences by studying random matrix models where nonasymptotic treatment of spectral properties is possible even near the ground energy. Still, we acknowledge that our model is nonlocal and perhaps deviates from local Hamiltonian problems in some aspects: the quantum easiness stems from the semicircular spectrum and does not directly explain why the low-energy problem of chaotic local Hamiltonians (whose spectral density has a tail instead) should also be easy. Nevertheless, we expect the following findings to extrapolate to local chaotic Hamiltonians: random matrix behavior can emerge from very few bits of randomness, and the spectrum is smooth and free of outliers (Figure~\ref{fig:good_bad}). 

Still, there is a wealth of quantum chaos phenomenology that requires formal treatment for quantum advantage implications. One direction is to show ETH (e.g., for the SYK models), which roughly means that nearby energy eigenstates are well connected to each other. We believe this can be formalized for the GUE, which should also extend to our Pauli string ensemble by the universality principle. Another direction is to reduce the locality of the Pauli string ensemble. In fact, our circuit complexity lower-bound argument remains nontrivial {even when the locality $k$ of each Hamiltonian term is reduced from $k = \Theta(n)$ to $k = \log(n)$,} which at least gives a hint of classical hardness. 

\acknowledgments
We thank Thomas Vidick, John Preskill, Anand Natarajan, John Wright, Gil Refael, Mehdi Soleimanifar, Andr\'as Gily\'en, Sam McArdle, William Kretschmer, Michael Kastoryano, Hsin-Yuan (Robert) Huang, Leo Zhou, and Nicola Pancotti for helpful discussions. CFC is supported by the Eddlemen Fellowship and AWS Center for Quantum Computing summer intern program. During this project MB was previously additionally affiliated with the AWS Center for Quantum Computing, Pasadena, USA.



\newpage

 \appendix

The remaining part of the work begins with proofs for the comparison principle (Section~\ref{sec:calculation_comparison}), including the moments and the resolvent. We instantiate the nonasymptotic properties of GUE in Section~\ref{sec:GUE}. The comparison results and GUE properties altogether allow us to calculate the properties of the Pauli string ensemble (Section~\ref{sec:apply_to_pauli}). In section~\ref{sec:circuit_lower_bounds}, we prove the circuit size lower bounds for the Pauli string ensemble, whose argument is independent of the comparison principle. 
Section~\ref{sec:missing_proofs} contains brief missing proofs. Section~\ref{sec:difficult_higher} contains an argument for why interpolation methods do not immediately exploit higher matching moments.



\section{Calculations for the Lindeberg principle}\label{sec:calculation_comparison}

In this section, we apply a version of the Lindeberg exchange principle for the $p$th moments and the resolvent moments. The main assumption we use is that two sums of independent matrices share the same lower-order moments. The main technical argument is readily illustrated in the moment calculation. The resolvent calculation is more involved because the resolvent is nonconvex.  We also have to establish concentration for a random realization of the resolvent moment around its expected value.

\subsection{Moments}\label{sec:moments}
We recapitulate the statement for moments.

\universalitymoments*

Our proof of Theorem~\ref{thm:universality_moments} is based on the Lindeberg exchange principle.  Roughly, we interpolate between the two sums by replacing one argument at each step.  Since the low moments of the summands match, each replacement only changes the $p$-norm slightly, with error on the order $(t+1)$. The calculation is straightforward, but it implicitly exploits noncommutativity properties of the random matrices in the moment matching.  The error is a noncommutative polynomial of matrices, and we treat them by a brutal application of H{\"o}lder's inequality, entirely ignoring noncommutativity.
Once we have replaced all the summands, we tie the estimates together using a self-bounding argument.  To execute this step, we must solve a difference equation by passing to a continuous differential equation.

Lindeberg's method has recently been applied to RMT in the papers~\cite{Cha05:Simple-Invariance,KM11:Applications-Lindeberg,OT18:Universality-Laws}.  Some of the other ideas and methods in this argument have roots in the papers~\cite{Tro16:Expected-Norm,Tro18:Second-Order-Matrix,Bandeira2021MatrixCI,Tatiana_22_universality}.

The basic argument relies on a general form of H{\"o}lder's inequality for Schatten norms.
As we will see, we can introduce more refined moment inequalities to obtain some improvements.

\begin{fact}[Multivariate H{\"o}lder for random matrices]\label{fact:products_holder}
For any family $(\vX_1,\dots,  \vX_{k})$ of square random matrices, possibly statistically dependent,
the product satisfies the trace inequality
\begin{align}
  \Expect  \btr \labs{\prod_{i=1}^{k} \vX_i} 
  = \vertiii{\prod_{i=1}^{k} \vX_i }_1 
  \le \prod_{i=1}^{k} \vertiii{\vX_i}_{p_i} \quad \text{whenever} \quad \sum_{i=1}^{k} \frac{1}{p_i} = 1
  \quad\text{and}\quad p_i \geq 0.
\end{align}

\end{fact}

\begin{proof}[Proof sketch]
In the deterministic setting, with two matrices, the result appears in~\cite[Corollary 4.2.6]{bhatia1997matrix}.  Use induction to extend the bound to more than two deterministic matrices.
To incorporate the normalized trace, note that the weighted geometric mean $(p_1, \dots, p_\ell) \mapsto \prod_{i=1}^\ell a_i^{p_i}$ is homogeneous for fixed $a_j \geq 0$.  To incorporate the expectation, recall that the weighted geometric mean is concave, and invoke Jensen's inequality.
\end{proof}


\begin{proof}[Proof of Theorem~\ref{thm:universality_moments}: Two matching moments]
To illustrate the concept behind the argument, we carefully establish the first bound for the $t = 2$ case.  Afterward, we describe the modifications required to extend
the bound to $t > 2$ and to introduce the variance parameter $\sigma^2$.

Fix an even natural number $p \in 2 \mathbb{N}$.  Our goal is to compare the $p$th moment $\vertiii{\cdot}_p^p$ of the two independent sums $\vS = \sum_{i=1}^m \vA_i$ and $\tvS = \sum_{i=1}^m \tvA_i$.
The main idea is to update one summand at a time from $\vA_j$ to $\tvA_j$,
controlling the change in the $p$-norm at each step.


In detail, for each index $j = 0, 1, 2, \dots, m$,
we can define the hybrid matrix
\[
\vS_j := \sum^j_{i=1} \tvA_i +  \sum^m_{i=j+1} \vA_i
\quad\text{where}\quad
\vS_0 = \vS
\quad\text{and}\quad
\vS_m = \tvS.
\]
Express the difference between the $p$th moments of $\vS$ and $\tvS$
as a telescoping sum:
\begin{align} \label{eqn:moment-recursion}
    \Expect \btr\left (\sum_{i=1}^m \tvA_i \right) ^p -\Expect \btr \left(\sum_{i=1}^m \vA_i \right)^p 
    &= \sum_{j=1}^{m} \left(\Expect \btr\vS_j^p - \Expect \btr\vS_{j-1}^p \right).
\end{align}
For even $p$, we can express the $p$-norm%
\footnote{The odd $p$-norms contain an absolute value $\vertiii{\vS}_p^p = \Expect \btr \labs{\vS}^p$, so they are less suitable for Taylor expansion.}
in terms of a trace power: $\vertiii{\vS}_p^p = \Expect \btr \vS^p$.
To bound the telescoping sum, we first analyze the single update error and then solve a recursion.

\textbf{Step I: Single update error.}
Fix an index $j = 1, \dots, m$.  Let us give a bound for the change in the $p$-norm when we update $\vA_j$ to $\tvA_j$.  Define the unchanged part of the sum $\vS_j$ by
\begin{align}
    \vS_-:= \vS_{j}-\tvA_{j} = \vS_{j-1}-\vA_{j}. 
\end{align}
We can control the change in the polynomial moment by performing a Taylor expansion of the polynomial moment at the unchanged part $\vS_-$.

When expanding powers of a sum of matrices, keep in mind the scalar binomial expansion:
\begin{align}
    (x+y)^p &= \sum_{k = 0}^{p} \binom{p}{k} x^{p-k}y^k. 
\end{align}
For matrices, the expansion takes the form
\begin{align}
    (\vS_- + \vA_{j})^p &= \sum_{words} \vS_- \cdots \vA_j \vS_-\cdots \vA_j\vS_- \cdots \vA_j \vS_-\cdots\\
    &= \sum_{k=0}^p \vM_k \quad \text{where}\quad  \labs{\Expect \btr \labs{\vM_k} } \le \binom{p}{k} \vertiii{\vS_{-}}_{p}^{p-k} \vertiii{\vA_j}_{p}^{k}. \label{eqn:holder-mk}
\end{align}
Likewise, $(\vS_- + \tvA_{j})^p = \sum_{i=1}^m \tilde{\vM}_i$ with analogous bounds for the summands. 

The bound in~\eqref{eqn:holder-mk} follows from Holder's inequality (Fact~\ref{fact:products_holder}), applied for all possible relative positions of $\vS_-$ and $\vA_j$ with all parameters $p_i = p$.  Note that this general bound ignores the noncommutativity of the matrices.  The binomial coefficients $\binom{p}{k}$ are exactly the number of relative positions of the $k$ appearances of $\vA_j$ among the $p-k$ appearances of $\vS_-$.

To proceed, since the order parameter $t = 2$, the first and second moments of the random matrices $\vA_i$ and $\tvA_i$ match for each index $i$.  As a consequence,
\begin{align}
    \Expect \btr [\vA_i \vB] &= \Expect \btr [\tvA_i \vB]; \\ 
    \Expect \btr [\vA_i \vB \vA_i \vC] &=\Expect \btr[\tvA_i \vB \tvA_i \vC] \quad \text{for arbitrary $\vB, \vC$ independent from $\vA_i, \tvA_i$}.
\end{align}
Crucially, this implies that subtracting the expected moments completely \textit{cancels} the first-order terms $\vM_1$ and $\tilde{\vM}_1$ and the second-order terms $\vM_2$ and $\tilde{\vM}_2$.  Thus,
\begin{align}
    \labs{ \Expect \btr\vS_j^p - \Expect \btr\vS_{j-1}^p} & =  \labs{\Expect \btr(\vS_- + \tvA_{j})^p - \Expect \btr(\vS_- + \vA_{j})^p }
    = \labs{ \sum_{k=3}^p (\vM_k - \tilde{\vM}_k) } \\
    %
    &\le \sum_{k=3}^p p^k \vertiii{\vS_{j-1}}_{p}^{p-k} \vertiii{\vA_{j}}_{p}^{k} \quad + \quad (\vA_j\mapsto \tvA_{j})\\
    & \le \frac{1}{4}\vertiii{\vS_{j-1}}_{p}^{p-3} \big(2p\vertiii{\vA_{j}}_{p}\big)^{3} +  \frac{1}{4}\big(2p\vertiii{\vA_{j}}_{p}\big)^{p} \quad + \quad (\vA_{j}\mapsto \tvA_{j}). 
    \label{eq:update}
\end{align}
The notation $(\vA_{j}\mapsto \tvA_{j})$ denotes a replica of the first term with $\vA_j$ replaced by $\tvA_j$.  To reach the second line, we collect the higher order terms $\vM_3,\cdots, \tilde{\vM}_p$ and $\tilde{\vM}_3,\cdots, \tilde{\vM}_p$, we bound the binomial coefficient $\binom{p}{k}\le p^k$, and we use the convexity of the $p$-norm $\vertiii{\vS_-}_{p} = \vertiii{\vS_-+\Expect [\vA_j]}_{p}\le \vertiii{\vS_-+\vA_j}_{p} = \vertiii{\vS_{j-1}}_{p}$. The last inequality bounds the geometric series using the elementary numerical inequality 
\begin{align}
    \sum_{k=3}^p {x^k} = \sum_{k=3}^p {2^{-k}(2x)^k} \le \frac{1}{4}( (2x)^3+(2x)^p)
    \quad\text{for $x \geq 0$.}
\end{align}
We have successfully established a comparison between the quantities $\vertiii{\vS_{j}}_{p}^p$ and $\vertiii{\vS_{j-1}}_{p}^p$.

\textbf{Step II: Solving the recursion.}
We have expressed the difference between the moments of $\vS$ and $\tvS$ as a telescoping sum~\eqref{eqn:moment-recursion} of moments of hybrid matrices $\vS_j$.  The first part of the argument yields a bound on the change in moments at each step in terms of a smaller moment of the hybrid matrices.  We can use these results to develop coupled difference inequalities, which we must solve.

Define the scalar quantities
\begin{align}
    x_j:=\Expect \btr \vS_j^p
    \quad\text{for $j = 0, 1, 2, \dots, m$.}
\end{align}
The boundary values of the sequence $(x_j)$
are the moments of the original independent sums that we seek to compare:
$x_{m} = \Expect \btr \tvS^p$ and $x_0=\Expect \btr \vS^p$.
To bound the differences of this sequence, we introduce the notation
\begin{align}
a_j :=\frac{(2p)^3}{4} \L( \vertiii{\vA_j}_{p}^{3} +\vertiiismall{\tvA_j}_{p}^{3}\R) \quad \text{and}\quad b_j := \frac{(2p)^p}{4} \L(\vertiii{\vA_j}_{p}^{p}+\vertiiismall{\tvA_j}_{p}^{p}\R)
\quad\text{and}\quad
a_0 := b_0 := 0.
\end{align}
Using the inequality~\eqref{eq:update} for the $j$th step of the exchange,
we arrive at the coupled difference inequalities:
\begin{align}
\labs{ x_{j}-x_{j-1}} \le a_j \cdot x_{j-1}^{(p - 3)/p} + b_j \quad \text{for each} \quad j = 1 ,\dots,  m. \label{eq:disc_recursion}
\end{align}
Our task is to produce bounds for the terminal value $x_m$ in terms of the initial value $x_0$ and the coefficients $a_j$ and $b_j$.

To do so, we pass to a differential equation.
The proof appears at the end of this section.


\begin{lem}[From differences to derivatives] \label{lem:diff_bounds_discrete}
Define coefficient functions
\[
a(s) := a_{\lceil s\rceil} \quad \text{and} \quad  b(s) := b_{\lceil s\rceil}
\quad \text{for} \quad s \in [0, m].
\]
For a fixed integer $1 \leq k \leq p$, consider the differential inequality
\begin{align} \label{eqn:diff-ineq}
\begin{cases}
    x'(s) \geq a(s) \cdot x (s)^{(p-k)/p} + b(s), & s \in [0,m]; \\
    x(0) = x_0.
    \end{cases}
\end{align}
Then each solution $x(s)$ to the differential inequality overestimates
the solution $(x_j)$ to the coupled difference inequalities~\eqref{eq:disc_recursion}
in the sense that
\begin{align}
    x (j) \ge x_{j} \ge 0 \quad \text{for each} \quad j = 0,\dots, m.
\end{align}
\end{lem}


The following ansatz provides
a solution to the differential inequality~\eqref{eqn:diff-ineq}.
The proof of this lemma appears at the end of the section.

\begin{lem}[Ansatz for differential inequality]\label{lem:ansatz}
Fix an integer $1 \le k \le p$.  Consider the function
\begin{align}
    y(s): = \L[ x(0)^{k/p} + \frac{k}{p}\int_0^s a(u) \idiff{u} + \L(\int_0^s b(u) \idiff{u} \R)^{k/p}\R]^{p/k}
    \quad\text{for $s \in [0,m]$.}
\end{align}
Then $y$ solves the differential inequality~\eqref{eqn:diff-ineq}.
\end{lem}

We are now prepared to solve the coupled difference inequalities.  Instantiate Lemma~\ref{lem:diff_bounds_discrete} and Lemma~\ref{lem:ansatz} with parameters $s=m$ and $k=3$ to arrive at the one-sided inequality
\begin{align}
x_{m}^{3/p} \le y(m)^{3/p} =  x_0^{3/p} +\frac{3}{p} \sum_{j=1}^m a_j+ \left(\sum_{j=1}^m b_j\right)^{3/p}. 
\end{align}
This inequality provides an upper bound for the difference $x_m^{3/p} - x_0^{3/p}$,
where we recall that $x_m$ and $x_0$ are the $p$th moments of the two independent sums.
Taking the third root and bounding the $\ell_3$ norm by the $\ell_1$ norm, we also
have the estimate
\begin{align}
x_{m}^{1/p} - x_0^{1/p} \leq \left( \frac{3}{p} \sum_{j=1}^m a_j \right)^{1/3} + \left(\sum_{j=1}^m b_j\right)^{1/p}.
\end{align}
This statement is slightly weaker, but it may be easier to interpret and apply.

%
We may repeat the same argument switching the roles of $\vA_j$ and $\tvA_j$, noting that coefficients $a_j$ and $b_j$ remain the same.  This yields the desired two-sided estimate:
\begin{align}
    \labs{\vertiii{\sum_{i=1}^m \vA_i}_p^{3}-\vertiii{\sum_{i=1}^m \tvA_i}_p^{3}}=\labs{x_m^{3/p} - x_0^{3/p}}\le \frac{3}{p} \sum_{j=1}^m a_j+ \left(\sum_{j=1}^m b_j \right)^{3/p}. 
\end{align}
Similarly,
\begin{align}
    \labs{\vertiii{\sum_{i=1}^m \vA_i}_p -\vertiii{\sum_{i=1}^m \tvA_i}_p}=\labs{x_m^{1/p} - x_0^{1/p}}\le \left( \frac{3}{p} \sum_{j=1}^m a_j \right)^{1/3} + \left(\sum_{j=1}^m b_j \right)^{1/p}. 
\end{align}
Introduce the values of $a_j$ and $b_j$ and evaluate the numerical constants to
complete the proof of the theorem for random matrix models with
matching second moments ($t=2$).
\end{proof}

\begin{proof}[Proof of Theorem~\ref{thm:universality_moments}: More matching moments]
Using an analogous argument, we can obtain related results comparing random matrix models where the moments match.  Fix $t \geq 2$.   Suppose that each pair $\vA_i$ and $\tvA_i$ of summands has matching moments up to order $t$.  In this case, the terms $\vM_{1}, \dots, \vM_{t}$ cancel with $\tilde{\vM}_1, \dots, \tilde{\vM}_t$, so the error depends only on the higher order terms $\vM_k$ and $\tilde{\vM}_k$ for $k \geq t+ 1$.

Pursuing this observation, we arrive at the bound
\begin{align}
      E \coloneqq \labs{\vertiii{\sum_{i=1}^m \vA_i}_p - \vertiii{\sum_{i=1}^m \tvA_i}_p} & \le \left(\frac{t+1}{p} \sum_{j=1}^m a_j \right)^{1/(t+1)} + \left(\sum_{j=1}^m b_j \right)^{1/p}.
\end{align}
where
\begin{align}
    \quad a_j :=\frac{(2p)^{t+1}}{2^{t}} \L( \vertiii{\vA_j}_{p}^{t+1} +\vertiiismall{\tvA_j}_{p}^{t+1}\R) \quad \text{and}\quad b_j := \frac{(2p)^p}{2^t} \L(\vertiii{\vA_j}_{p}^{p}+\vertiiismall{\tvA_j}_{p}^{p}\R).
\end{align}
Using the notation $L_{p, k}$ from the statement of the theorem, we reach
the estimate
\[
E \leq (2 (t+1))^{1/(t+1)} \cdot p^{t/(t+1)} \cdot L_{p, t+1}
    + 2^{1 - t/p} \cdot p \cdot L_{p,p}.
\]
For $2 \leq t \leq p$, each of the leading constants is bounded above by $2$.
This completes the argument.
\end{proof}

\begin{proof}[Proof of Theorem~\ref{thm:universality_moments}: Refined statistics]
Last, we establish the result with more precise statistics of the
random summands.  To do so, we simply replace H{\"o}lder's inequality
(Fact~\ref{fact:products_holder}) by a more refined moment inequality.
Here is the statement we require,
which specializes~\cite[Prop.~4.1]{Tatiana_22_universality}.

\begin{fact}[Trace inequality for random matrices]
Let $\vA$ and $\vY$ be random Hermitian matrices that are statistically independent.
Consider a product with $k$ copies of $\vA$ and $(p-k)$ copies of $\vY$ in any order,
where $p \geq k \geq 2$.  Then
\[
\Expect \btr [ \vA \vY^{p_1} \cdots \vA \vY^{p_k}]
    \leq \left[ \norm{ \Expect \vA^2 } \cdot \vertiii{ \vA }_{\infty}^{k-2} \right]
    \cdot \vertiii{ \vY }_{p}^{p-k}.
\]
In this expression, $p_1 + \dots + p_k = p - k$ is an integer partition.
\end{fact}

We invoke this result with $\vY = \vS_{-}$ and with $\vA = \vA_j$
or $\vA = \tvA_j$.  We can also take the minimum of this bound with
the bound via H{\"o}lder's inequality to see that the tail parameter
$L_{p,p}$ does not get worse.  The rest of the proof is the same.
\end{proof}

Finally, we complete the proofs of the two lemmas that
were required in the argument.


\begin{proof}[Proof of Lemma~\ref{lem:diff_bounds_discrete}]
By assumption $x(0) \geq x_0 = 0$.
For an induction, assume that $x (j-1) \ge x_{j-1}$ for an index $j \geq 1$.
Then
\begin{align}
    x(j)-x(j-1) = \int_{j-1}^j x'(s) \idiff{s}
    &\geq \int_{j-1}^{j} \L(a(s) \cdot x (s)^{\frac{p-k}{p}} + b(s) \R) \idiff{s} \\
    &\ge \int_{j-1}^{j} \L(a(j) \cdot x (j-1)^{\frac{p-k}{p}} + b(j)\R)\idiff{s} \ge  x_j- x_{j-1}.
\end{align}
To reach the second line, note that the function $x(s)$ is increasing because the right-hand side of the differential inequality is positive.  Then observe that the coefficients $a(s) = a(\lceil s\rceil)$ and $b(s) = b(\lceil s\rceil)$ are constant on the domain of integration.   By induction, we obtain
\begin{align}
    x (j) \ge x_{j} \quad \text{for each} \quad j = 0,\dots, m.
\end{align}
This is the stated result.
\end{proof}

\begin{proof}[Proof of Lemma~\ref{lem:ansatz}]
To verify that the ansatz satisfies the differential inequality,
first note the initial condition $y(0) = x(0)$.
Using this fact, we take the derivative:
\begin{align}
    \frac{\diff{y}(s)}{\diff{s}} &= \frac{p}{k}\L[ y(0)^{k/p} + \frac{k}{p}\int_0^s a(u) \idiff{u} + \L(\int_0^s b(u) \idiff{u} \R)^\frac{k}{p}\R]^{(p/k)-1}  \L( \frac{k}{p}a(s) + \L( \int_0^s b(u) \idiff{u} \R)^{(k/p)-1}\cdot \frac{k}{p} b(s) \R)\\
    &\ge a(s) \L[ y(0)^{k/p} + \frac{k}{p}\int_0^s a(u)\idiff{u} + \L(\int_0^s b(u) \idiff{u} \R)^{k/p} \R]^{(p-k)/k} + b(s)\\
    &= a(s) y(s)^{(p-k)/p} + b(s).
\end{align}
The inequality depends on the fact that $(k/p) - 1 \le 0$ and 
\begin{align}
   \L( \int_0^s b(u) du \R)^{(k/p) -1} \ge \L[ y(0)^{k/p} + \frac{k}{p}\int_0^s a(u) \diff{u} + \L(\int_0^s b(u) \diff{u} \R)^{k/p}\R]^{1-(p/k)}. 
\end{align}
This is a direct calculation using the fact that all the terms are positive.
Therefore, the ansatz solves the differential inequality.
%
\end{proof}

\subsection{The resolvent}\label{sec:resolvent}

The comparison principle extends to other functions besides polynomial moments. In this section, we study moments of the resolvent:
\begin{align}
    \vR: = \frac{1}{\vH - \omega +\ri \eta}\quad \text{and}\quad \tvR: = \frac{1}{\tvH - \omega +\ri \eta}.
\end{align}
As usual, $\vH$ and $\tvH$ are defined in~\eqref{eqn:rdm-sum-intro}.  The parameters $\omega \in \BR$ and $\eta > 0$.

\universalityResolvent*

Whenever the right-hand side is small ($\ll (p\eta)^{-1}$), we may take $p$th power to obtain the expected density of states (filtered by the resolvent) up to a \textit{multiplicative} error. For our Pauli string ensemble, we can achieve this outcome because $L_{3p,\infty} = 1/\sqrt{m}$, and the number $m$ of summands is chosen sufficiently large.

As compared with the polynomial moments, universality for the resolvent involves some additional technical challenges.  They stem from the fact that the resolvent has an infinite Taylor series, and it is a nonconvex function of the random matrix.
To address the first concern, we follow~\cite{Tatiana_22_universality} and truncate the Taylor series at a carefully chosen order.  To that end, let us recall the statement of Taylor's theorem with an integral remainder.

\begin{fact}[Taylor with integral remainder]
\label{fact:Taylor}
If the function $f : [0,1] \to \BR$ is $K$ times continuously differentiable, then
\begin{align}
     f(1) = \sum_{k=0}^{K - 1} \frac{f^{(k)}(0)}{k!} + \frac{1}{(K-1)!} \left[ \int_{0}^{1}f^{(K)}(s) (1-s)^{K-1} \idiff{s} \right].
\end{align}
\end{fact}

The Taylor expansion of the resolvent has a rather involved expression. Fortunately, we merely need bounds for the higher-order terms. 

\begin{prop}[Expanding the resolvent] \label{prop:expand_resolvent}
For Hermitian matrices $\vS$ and $\vA$ of the same order, consider the matrix $\vZ = \vS + \iunit\eta \vI$ where $\eta \in \BR$. Then, for each even natural number $p \in 2 \mathbb{N}$,
\begin{align}
\frac{1}{\labs{\vZ+\vA}^p} 
&= \sum_{k=0}^{3p} \vM_k \quad \text{where}\quad \Expect\btr \labs{\vM_k} 
    \le \frac{(4p)^{k}}{\eta^{4k/3}} \cdot \vertiii{\vZ^{-1}}_{p}^{p-k/3} \vertiiismall{\vA}^k_{3p}\quad \text{for $k = 0,\dots, 3p$.}
    \label{eq:resolvent_expansion}
\end{align}
The term $\vM_k$ is a noncommutative polynomial of degree $k$ in the variable $\vA$
and degree $p + k$ in the variables $\vZ^{-1}$ and $\vZ^{-\dagger}$, where ${}^{-\dagger}$
refers to the conjugate transpose of the inverse.
\end{prop}
\begin{proof}[Proof of Proposition~\ref{prop:expand_resolvent}]
First, we expand 
\begin{align}
   \frac{1}{\labs{\vZ+\vA}^p}&=\frac{1}{(\vZ+\vA)^{p/2}} \times \frac{1}{((\vZ+\vA)^{\dagger})^{p/2}}\\
   &= \sum_{\ell_1=0}^{\infty} (-\vZ^{-1}\vA)^{\ell_1}\vZ^{-1}\cdots \sum_{\ell_{p/2}=0}^{\infty} (-\vZ^{-1}\vA)^{\ell_{p/2}}\vZ^{-1}\\ &\qquad\times \sum_{\ell'_1=0}^{\infty} (-\vZ^{-\dagger}\vA)^{\ell'_1}\vZ^{-\dagger}\cdots \sum_{\ell'_{p/2}=0}^{\infty} (-\vZ^{-\dagger}\vA)^{\ell'_{p/2}}\vZ^{-\dagger}.\label{eq:ZplusAp}
\end{align}
The first line depends on the fact that $\vZ + \vA$ and $\vZ^{\dagger} + \vA$ commute. The second line uses the expansion $(\vZ + \vA)^{-1} = \sum_{\ell=0}^{\infty} (-\vZ^{-1}\vA)^{\ell}\vZ^{-1}$.

Next, we collect into the matrix $\vM_k$ all terms with total power $k$ on the matrix $\vA$ and total power $p + k$ on the matrix $\vZ^{-1}$ or $\vZ^{-\dagger}$.
For $ 0 \le k < 3p $, there are $\binom{p+k-1}{k}\le (4p)^k$  such terms.  Then we apply H{\"o}lder's inequality (Fact~\ref{fact:products_holder}) to each term contributing to $\vM_k$.  This step yields  
\begin{align}
    \Expect\tr \labs{\vM_k} \le (4p)^k \cdot \vertiii{\vZ^{-1}}^{p+k}_{q_1} \vertiiismall{\vA}_{q_2}^{k} \quad \text{for}\quad  q_1= \frac{3p}{3p-k} \quad \text{and} \quad q_2=\frac{3p}{k}.
\end{align}
To bring the bounds into the same form, note that $\vertiii{\smash{\vZ^{-1}}}_{q_1} \le \vertiii{\smash{\vZ^{-1}}}_{p}$ for $q_1 \leq p$ and $\vertiii{\vA}_{q} \le \vertiii{\vA}_{3p}$ for $q_2 \le 3p$.   To treat additional powers of $\vZ^{-1}$, use a uniform bound for the resolvent: $\vertiii{\smash{\vZ^{-1}}}_{q_1}^{4k/3} \leq \norm{\vZ^{-1}}^{4k/3} \le \eta^{-4k/3}$.  We reach the advertised bound for terms of order $0\le k < 3p$.

For the remainder term in the Taylor expansion ($K = 3p$), we may compute the $K$th derivative  using~\eqref{eq:ZplusAp} and invoke the same method to obtain a bound.
\begin{align}
     \labs{\frac{\diff{}^{3p}}{\diff{s}^{3p}} \BE \btr\frac{1}{\labs{\vZ+s\vA}^p}} \le
     (3p)! \cdot (4p)^{3p} \cdot \frac{\vertiiismall{\vA}^{3p}_{3p}}{\eta^{3p}}
    \quad\text{for $s \in [0,1]$.}
\end{align}
We have applied the uniform bound $\norm{(\vZ +s\vA)^{-1}} \le \eta^{-1}$.  Introduce the last display into the integral remainder term in the Taylor expansion (\ref{fact:Taylor}).  We reach the required estimate for $K=3p$.
\end{proof}

\begin{proof}[Proof of Theorem~\ref{thm:universality_resolvent}]
To obtain a comparison of the resolvents, we apply Lindeberg's method again.
For clarity of argument, we will assume that there are $t = 2$ matching moments;
the general case is similar.  Define hybrid matrices and their resolvents:
\begin{align}
\vH_j := \sum^j_{i=1} \tvA_i +  \sum^m_{j+1} \vA_i 
\quad\text{and}\quad
\vR_j := \frac{1}{\vH_j - \omega +\ri \eta}
\quad\text{for $j = 0, \dots, m$.}
\end{align}
Consider the telescoping sum
\begin{align}
    \Expect \btr\abs{\tvR}^p - \Expect \btr\labs{\vR}^p&= \sum_{j = 1}^m \big(\Expect \btr \labs{\vR_j}^p - \Expect \btr\labs{\vR_{j-1}}^p \big). 
\end{align}
We must bound each of the terms in the telescope.

\textbf{Step I: Single update error.}
Fix an index $j = 1, \dots, m.$  The $j$th update replaces the summand $\vA_j$ with $\tvA_j$.
Define the unchanged part of the matrix and its resolvent:
\begin{align}
    \vH_-:= \vH_j-\tvA_j= \vH_{j-1}-\vA_j 
    \quad\text{and}\quad
    \vR_- := (\vH_- - \omega +\ri \eta)^{-1}.
\end{align}
Since the moments of $\vA_j$ and $\tvA_j$ match up to second order and these matrices are independent from $\vH_{-}$, the terms $\vM_0, \vM_1, \vM_2$ cancel the terms $\tilde{\vM}_0, \tilde{\vM}_1, \tilde{\vM_2}$
in the Taylor expansion of the resolvent powers (Proposition~\ref{prop:expand_resolvent}).
Thus,
\begin{align}
    \labs{  \Expect \btr \labs{\vR_j}^p - \Expect \btr \labs{\vR_{j-1}}^p } & =  \labs{\Expect \btr \labs{(\vH_--\omega +\ri \eta) +\vA_j }^{-p} - \Expect\btr\labs{\smash{(\vH_- -\omega +\ri \eta) + \tvA_j}}^{-p} }\\
    &\le \frac{1}{4} \vertiii{\vR_{-}}_{p}^{p-1} \frac{(8p)^{3} \vertiiismall{\tvA_j}_{3p}^{3}}{\eta^{4}} + \frac{1}{4}   \frac{(8p)^{3p}\vertiiismall{\tvA_j}_{3p}^{3p}}{\eta^{4p}} \quad+\quad (\tvA_j \mapsto \vA_j)\\
    &\le \frac{1}{2} \vertiii{\vR_{-}}_{p}^{p-1} \frac{(8p c_j)^{3}}{\eta^{4}} + \frac{1}{2}  \frac{(8pc_j)^{3p}}{\eta^{4p}} \quad \quad \text{(setting $c_j := \max\{ \vertiiismall{\vA_j}_{3p}, \vertiii{\smash{\tvA_j}}_{3p} \}$)}\\
    &\le \vertiii{\vR_{j-1}}_{p}^{p-1} \frac{(8pc_j)^{3}}{\eta^{4}} + \frac{3}{2}  \frac{(8pc_j)^{3p}}{\eta^{4p}} .\label{eq:resolvent_difference}
\end{align}
The first inequality bounds the geometric series of the error terms~\eqref{eq:resolvent_expansion}:
\begin{align}
    \sum_{k=3}^{3p} {x^k}
    \le \frac{1}{4}( (2x)^3+(2x)^{3p})\quad \text{for} \quad x:= \frac{4p}{\eta^{4/3}} \cdot \vertiii{\vR_{-}}_{p}^{-1/3} \vertiiismall{\tvA_j}_{3p}.
\end{align}
The second inequality combines the bounds for the two different summands $\vA_j$ and $\tvA_j$.

The third inequality requires some comment.
By another Taylor expansion, we may control the moments of $\vR_{-}$ using the moments of $\vR_{j-1}$:
\begin{align}
    \Expect \btr\labs{\vR_{-}}^p &\le
    \Expect \btr \labs{\vR_{j-1}}^p \left(\sum_{k=0}^{3p} y^{k}\right) 
    \quad\text{where}\quad
    y:= \frac{4p}{\eta^{4/3}} \cdot \vertiii{\vR_{j-1}}_{p}^{-1/3} \vertiiismall{\vA_j}_{3p}.
\end{align}
Bounding the geometric series as $\sum_{k=0}^{3p} {y^k} \le 2( 1+((2y)^{3p})^{p/(p-1)})$
and noting that $\vertiiismall{\vA_j}_{3p} \le c_j$, we find that
\begin{align}
    \Expect \btr\labs{\vR_{-}}^p
         \le 2\L( \Expect \btr \labs{\vR_{j-1}}^p + \frac{ (8pc_j)^{3p}}{\eta^{4p}} \R)
\end{align}
Last, raise both sides to the $(p-1)/p$ power, and use the numerical inequality
$(a+b)^{(p-1)/p} \le a^{(p-1)/p}+b^{(p-1)/p}$ for $a,b \ge 0$ to reach~\eqref{eq:resolvent_difference}.
A similar bound holds when $t \geq 2$ moments match.

\textbf{Step II: Solving the recursion.} 
The recursion is similar to the proof of Theorem \ref{thm:universality_moments}.
We will present this argument for a general choice of $t \geq 2$.
First, introduce the scalar variables
$$
x_j:=\Expect \btr \labs{\vR_j}^p
\quad\text{for $j = 0, \dots, m$.}
$$
Define the coefficients $a_0 := b_0 := 0$ and
\begin{align}
    a_j :=\frac{1}{2^{t-2}} \frac{(2(t+2)pc_j)^{t+1}}{\eta^{t+2}} \quad \text{and}\quad b_j := \frac{3}{2^{t-1}} \frac{(2(t+2)pc_j)^{(t+1)p}}{\eta^{(t+2)p}}.
\end{align}
The updates~\eqref{eq:resolvent_difference} can then be written as a scalar recursion
\begin{align}
     \labs{ x_{j}-x_{j-1}} \le a_j \cdot x_{j-1}^{(p-1)/p} + b_j \quad \text{for each} \quad j = 1 ,\dots,  m. 
\end{align}
Repeating the same arguments as before (via Lemma~\ref{lem:diff_bounds_discrete} and Lemma~\ref{lem:ansatz}), we obtain control on the endpoints of the sequence:
\begin{align}
\labs{x_m^{1/p} - x_0^{1/p}} &\le \frac{1}{p} \sum_{j=1}^m a_j+ \left(\sum_{j=1}^m b_j \right)^{1/p}\\
&\le 8m\frac{p^t ((t+2)L_{3p,\infty})^{t+1}}{\eta^{t+2}} + \left(\frac{3}{2}\right)^{1/p} m^{1/p} \frac{(2(t+2)pL_{3p,\infty})^{t+1}}{2^{(t-2)/p}\eta^{t+2}} \\
&\le (2^{t+2}+10m/p) \frac{((t+2) p L_{3p,\infty})^{t+1}}{\eta^{t+2}}.
\end{align}
The second inequality uses the uniform bound $c_j = \max\{\vertiiismall{\tvA_j}_{3p},\vertiii{\vA_j}_{3p}\} \le L_{3p,\infty}$.
To reach the last line, note that $(3/2)^{1/p}\le 2$, drop the denominator $2^{(t-2)/p} \ge 1$,
and apply Young's inequality to determine that $m^{1/p}\le 1+m/p$.
\end{proof}
\subsection{Concentration for resolvent trace}

In this section, we study the $q$th moments of the resolvent trace, which, by Markov's inequality, gives the concentration of local density of states needed for Theorem~\ref{thm:low-energy-density}. The concentration fundamentally differs from the calculation for the expectation and does not explicitly refer to an ideal random matrix ensemble (e.g., the GUE). It suffices to introduce an independent copy by the convexity of $q$-norm 
\begin{align}
    \labs{\btr\labs{\vR}^p-\BE \btr\labs{\vR}^p }_q \le \labs{\btr\labs{\vR}^p-\btr\labs{\vR'}^p }_q \quad \text{where}\quad \labs{x}_q: = (\Expect \labs{x}^q)^{1/q} 
\end{align}
which allows us to utilize powerful concentration inequality for martingales. The estimate depends on an expected moment $\labs{\btr\labs{\vR}^p}_q$, which we bound independently in Section~\ref{sec:resolvent_expected_moment}.

\begin{thm}[Concentration for resolvent trace]\label{thm:resolvent_con_good_lambda}
 For independent centered matrices $\vA_1,\dots, \vA_m$,
consider identical copies $\vA'_j$ of $\vA_j$. Then, the resolvent trace concentrates
    \begin{align}
    \labs{\btr\labs{\vR}^p-\btr\labs{\vR'}^p}_q
    \lesssim \labs{\btr\labs{\vR}^p}_q \L( \frac{\sqrt{q} p^2 }{\eta^2}\sqrt{\sum_{j=1}^m \labs{\norm{\vA_{j}}}_{\infty}^4}+ \frac{\sqrt{q}p}{\eta^2}\sqrt{\sum^m_{j=1} 
    \sigma_{*}(\vA_j)^2} + \frac{qp}{\eta}\left(\sum_{j=1}^m \labs{\norm{\vA_{j}}}^q_{\infty}\right)^{1/q}\R)
\end{align}
where
\begin{align}
    \sigma_{*}(\vA)^2 := \sup_{\norm{\ket{u}}=\norm{\ket{v}} = 1} \BE_{\vA} \labs{\bra{u}\vA\ket{v}}^2.
\end{align}
\end{thm}

Crucially, the estimate depends on a variance-like quantity $\sigma_*^2(\vA)$ that more fully reflects the randomness of the random matrix $\vA$. For the Pauli string ensemble, this quantity is significantly smaller than the ordinary matrix variance: $\sigma_*^2(\vA)= (mN)^{-1} \ll m^{-1} =\Expect \labs{\norm{\vA_j}}^2$. The quantity $\sigma_*^2(\vA)$ arises from the following bound.

\begin{fact}\label{fact:sigmastar}
Consider a random matrix $\vA$ and a fixed matrix $\vB$ with compatible dimensions, then
\begin{align}
    \BE_{\vA} \labs{\tr[\vA\vB]}^2 \le \sigma_*^2(\vA) \cdot \tr[\labs{\vB}]^2 
\end{align}
\end{fact}
\begin{proof}[Proof of Fact~\ref{fact:sigmastar}]
    Consider the singular value decomposition 
     $\vB = \sum_j \ket{v_j} s_j\bra{u_j}$. 
    Then,
    \begin{align}
        \BE_{\vA} {\tr[\vA\vB]\tr[\vA^{\dagger}\vB^{\dagger}]} &=  \BE_{\vA}\sum_{j}\bra{u_j}\vA\ket{v_j} s_j  \sum_j\bra{v_j}\vA^{\dagger}\ket{u_j} s_j\\
        &\le \sup_{j,i} \BE_{\vA} 
        \labs{\bra{u_j}\vA\ket{v_j}} \labs{\bra{u_i}\vA\ket{v_i}} \cdot \tr[\labs{\vB}]\tr[\labs{\vB}]\\
        &\le \sup_{\norm{\ket{u}}=\norm{\ket{v}} = 1} \BE_{\vA} \labs{\bra{u}\vA\ket{v}}^2 \cdot \tr[\labs{\vB}]\tr[\labs{\vB}].
    \end{align}
    The first inequality pushes the expectation inside the sum, and it applies H{\"o}lder's inequality to the sum. The second inequality is Cauchy--Schwarz. This is the advertised result.
\end{proof}

Also, the proof of Theorem~\ref{thm:resolvent_con_good_lambda} employs a refined scalar martingale inequality as follows. (For an introduction to martingales, see~\cite{williams_1991}.)

\begin{thm}[{\cite{hit_good_lambda}}]
\label{thm:good-lambda-moments} For scalar martingales difference sequence $d_j$ (i.e., $\Expect_{j-1}d_j:=\Expect[d_j|d_{j-1},\cdots, d_1]= 0)$, we have that 
\begin{align}
\left| \max_{k\le n} \labs{\sum^k_{j=1} d_j} \right|_q &\lesssim \left(\sqrt{q}\left| \sum^n_{j=1} \BE_{j-1}d^*_jd_j \right|^{1/2}_{\frac{q}{2}} +q\left|\max_{1\le j\le n} |d|_j \right|_q\right).
\end{align}
\end{thm}

Significantly, the conditional expectation $\Expect_{j-1}$ appears inside the norm, which then allows us to exploit the second-moment properties of $\vA_i$ via Fact~\ref{fact:sigmastar}. Otherwise, applying a crude martingale inequality, such as uniform smoothness, gives a looser bound in terms of $\Expect \labs{\norm{\vA_j}}^2$ instead of $\sigma_*^2(\vA)$.  The weaker bound does not properly reflect the randomness in $\vA_j$.

\begin{proof}[Proof of Theorem~\ref{thm:resolvent_con_good_lambda}]
As usual, we write the telescoping sum
\begin{align}
    \btr\abs{\vR'}^p - \btr\labs{\vR}^p&= \sum_{j = 1}^m (\btr \labs{\vR_j}^p - \btr\labs{\vR_{j-1}}^p ) \quad
\text{with}\quad
    \vR_j := \frac{1}{\vH_j - \omega +\ri \eta}\quad \text{and}\quad  \vH_j := \sum^j_{i=1} \vA'_i +  \sum^m_{j+1} \vA_i.
\end{align}
By construction, the updates compose a martingale difference sequence:
\begin{align}
    \Expect\L[ \btr\labs{\vR_{j}}^p-\btr\labs{\vR_{j-1}}^p \mid \vA'_{j-1},\vA_{j-1},\dots, \vA'_1,\vA_1\R] = 0 \label{eq:R'jRj}.
\end{align}
This point is evident because we can swap the random variables $\vA_j$ and $\vA_j'$ without changing the distribution.  To analyze the martingale, we expand the difference using the algebraic identity for a difference of powers:
\begin{align}
\btr\labs{\vR_{j}}^p-\btr\labs{\vR_{j-1}}^p & = \btr\L[ \vR_{j}^{\dagger}\vR_{j}\labs{\vR_{j}}^{p-2}  - \vR_{j-1}^{\dagger}\vR_{j-1}\labs{\vR_{j-1}}^{p-2}\R]\\
&= \btr\L[ (\vR^{\dagger}_{j} - \vR^{\dagger}_{j-1}) \cdot \vR_{j-1}\labs{\vR_{j-1}}^{p-2}\R] + \cdots\\
&= \btr\L[ \vR^{\dagger}_{j}(\vH_{j-1}-\vH_{j}) \vR_{j-1}^{\dagger}\cdot \vR_{j-1}\labs{\vR_{j-1}}^{p-2}\R] + \cdots\\
&= \btr\L[ \vR^{\dagger}_{j}(\vA_{j}-\vA'_{j}) \labs{\vR_{j-1}}^{p}\R] + \cdots\\
& = a_j - b_j + c_j. 
\end{align}
The second equality is a telescoping sum.  For the moment, we have suppressed other telescoping terms, such as $\btr\L[ \vR_j^{\dagger}(\vR_{j} - \vR_{j-1}) \cdot \labs{\vR_{j-1}}^{p-2}\R]$. The third equality uses the matrix identity $\vA^{-1} - \vB^{-1} = \vB^{-1} (\vB-\vA)\vA^{-1}$. The last line regroups into three types of terms
\begin{align}
    a_j&:= \btr\L[ \vR^{\dagger}_{j} \vA_{j} \labs{\vR_{j}}^{p}\R] +\btr\L[ \labs{\vR}_{j}^2 \vA_{j} \vR_{j}\labs{\vR_{j}}^{p-2}\R]+\cdots =:  \sum_{\substack{r+s = p+1,\\ r,s\ge 1} }\btr\L[ \vR^{(r)}_{j} \vA_{j} \vR_{j}^{(s)}\R]\\
    b_j&:= \btr\L[ \vR^{\dagger}_{j-1} \vA'_{j} \labs{\vR_{j-1}}^{p}\R]+ \btr\L[ \labs{\vR}_{j-1}^2 \vA_{j}' \vR_{j-1}\labs{\vR_{j-1}}^{p-2}\R]+\cdots =: \sum_{\substack{r+s = p+1,\\ r,s\ge 1} }\btr\L[ \vR^{(r)}_{j-1} \vA'_{j} \vR_{j-1}^{(s)}\R]\\ 
    c_j&:= \btr\L[ \vR^{\dagger}_{j} \vA_{j} (\labs{\vR_{j-1}}^{p} - \labs{\vR_{j}}^{p})\R]- \btr\L[ (\vR^{\dagger}_{j} - \vR^{\dagger}_{j-1}) \vA'_{j} \labs{\vR_{j-1}}^{p}\R]+\cdots\\
    &=:\sum_{\substack{r+s = p+1,\\ r,s\ge 1} }\btr\L[ \vR^{(r)}_{j} \vA_{j} (\vR_{j-1}^{(s)}-\vR_{j}^{(s)})\R] -  \btr\L[ (\vR^{(r)}_{j}-\vR^{(r)}_{j-1}) \vA'_{j} \vR_{j-1}^{(s)}\R]
\end{align}
using the identities for each $r$ and $s$
\begin{align}
    \btr\L[ \vR^{(r)}_{j} \vA_{j} \vR_{j-1}^{(s)}\R] &= \btr\L[ \vR^{(r)}_{j} \vA_{j} \vR_{j}^{(s)}\R] + \btr\L[ \vR^{(r)}_{j} \vA_{j} (\vR_{j-1}^{(s)}-\vR_{j}^{(s)})\R]\\
    \btr\L[ \vR^{(r)}_{j} \vA'_{j} \vR_{j-1}^{(s)}\R] &= \btr\L[ \vR^{(r)}_{j-1} \vA'_{j} \vR_{j-1}^{(s)}\R] + \btr\L[ (\vR^{(r)}_{j}-\vR^{(r)}_{j-1}) \vA'_{j} \vR_{j-1}^{(s)}\R].
\end{align}
To simplify the expression above, we defined
\begin{align}
    \vR_j^{(r)} := \begin{cases}
    \vR_j^{\dagger} \labs{\vR_j}^{r-1} \quad \text{or}\quad \vR_j\labs{\vR_j}^{r-1} &\text{if} \quad r \quad \text{is odd}\\
    \labs{\vR}_j^{r}  & \text{if} \quad r \quad \text{is even}.
    \end{cases}
\end{align}
When $r$ is odd, we overload the same symbol for the above two possible types of expressions, both of which fit into the same argument since we eventually take norms everywhere. Note that $\labs{\vR}_j^{r-1} \vR^{\dagger}_j = \vR^{\dagger}_j \labs{\vR}^{r-1}_j$ since they can be simultaneously diagonalized in the eigenbasis of $\vH_j$. 

It remains to study concentration for each sequence $a_j,b_j,c_j$. The main observation is that each sequence forms a martingale difference sequence
\begin{align}
    \Expect\L[ a_j,b_j, c_j \mid \vA'_{j-1},\vA_{j-1},\dots, \vA'_1,\vA_1\R] = 0.
\end{align}
Indeed, the sequence $a_j$ gives a martingale difference sequence since $\vA_j$ is independent of $\vR_{j}$ (similarly for the sequence $b_j$). Even though complicated, the sequence $c_j$ also gives a martingale difference sequence since $c_j = (\btr\labs{\vR_{j}}^p-\btr\labs{\vR_{j-1}}^p) - a_j-b_j$ and~\eqref{eq:R'jRj}.
Intuitively, the sequences $a_j,b_j$ enjoy concentration controlled by the second-moments, which fully exploits randomness in the random matrix $\vA_j$ via Fact~\ref{fact:sigmastar}. The sequence $c_j$ is higher-order in $\vA_j$ and $\vA_j'$, and we simply use crude inequalities.

By uniform smoothness, the martingale difference sequence $c_j$ satisfies
\begin{align}
\labs{\sum_{j=1}^m c_j}_q^2 &\le \sum_{j=1}^m  (q-1)\labs{c_j}_q^2 \label{eq:cjBound}.
\end{align}
We evaluate the individual $q$-norms
\begin{align}
\labs{ c_j}_q   &\le  \sum_{\substack{r+s = p+1,\\ r,s\ge 1} } \Bigg( \labs{ \btr\L[ \vR^{(r)}_{j} \vA_{j} (\vR_{j-1}^{(s)}-\vR_{j}^{(s)})\R]}_q + \labs{\btr\L[ (\vR^{(r)}_{j}-\vR^{(r)}_{j-1}) \vA'_{j} \vR_{j-1}^{(s)}\R]}_q \Bigg) \\
& \le  \sum_{\substack{r+s+t = p+2,\\ r,s,t\ge 1} } \Bigg( \labs{ \btr\L[ \vR^{(r)}_{j} \vA_{j} \vR^{(s)}_{j}(\vA'_j-\vA_j) \vR^{(t)}_{j-1}\R] }_q + \labs{\btr\L[ \vR^{(r)}_{j}(\vA_j-\vA_j') \vR^{(s)}_{j-1} \vA'_{j} \vR_{j-1}^{(t)}\R] }_q \Bigg) \\
& \le  \sum_{\substack{r+s+t = p+2,\\ r,s,t\ge 1} } \Bigg( \labs{ \norm{\vA_{j}}\norm{\vA'_j-\vA_{j}} \norm{\vR_j}_q^{r/q}  \norm{\vR_j}_q^{s/q} \norm{\vR_{j-1}}_q^{t/q}   }_q \\
&\qquad\qquad\qquad\qquad + \labs{\norm{\vA'_{j}}\norm{\vA'_j-\vA_{j}} \norm{\vR_j}_q^{r/q}  \norm{\vR_{j-1}}_q^{s/q} \norm{\vR_{j-1}}_q^{t/q} }_q \Bigg)\\
&\le 4 \cdot \frac{p(p+1)}{2} \frac{ \labs{\norm{\vA_{j}}}_{\infty}^2}{\eta^2}\labs{\btr\labs{\vR}^p}_q.
\end{align}	
The second inequality further expands the resolvent difference. The third inequality uses Holder's inequality for the trace. The last inequality uses Holder's inequality, counts the combination of $r,s,t$ by $\binom{p+1}{2}$, uses that $\vA$ and $\vA'$ have the same distributions, and uses that $\vR_j$, $\vR_{j-1}$ have the same distribution as $\vR$.

For the sequence $b_j$, we apply the scalar martingale inequality (Theorem~\ref{thm:good-lambda-moments}). We calculate the predictable quadratic variation
\begin{align}
    \left| \sum^m_{j=1} \BE_{j-1}b^*_jb_j \right|_{\frac{q}{2}} &\le \sum_{\substack{r+s = p+1,\\ r,s\ge 1} } \labs{\sum^m_{j=1}\btr\L[ \vR^{(r)}_{j-1} \vA'_{j} \vR_{j-1}^{(s)}\R]^2}_q\\
    &\le  p^2\left| \sum^m_{j=1} 
    \sup_{\norm{\ket{u}}=\norm{\ket{v}} = 1} \BE_{\vA_j} \labs{\bra{u}\vA_j\ket{v}}^2 \cdot \btr\L[\labs{\vR_{j-1}}^{p+1}\R]^2 \right|_{\frac{q}{2}}\\
    &\le \frac{p^2}{\eta^2} \left|\btr\L[\labs{\vR}^{p}\R]\right|_{q}^2 \cdot  \sum^m_{j=1} \sigma_{*}(\vA_j)^2
\end{align}
using Fact~\ref{fact:sigmastar} and the uniform bound on resolvent $\norm{\vR} \le \eta^{-1}$. 
We calculate the maximum by 
\begin{align}
    \left|\max_{1\le j\le m} |b_j| \right|_q &\le \L(\sum_{j=1}^m\left|b_j\right|_q^q\R)^{1/q} \le p \frac{\L(\sum_{j=1}^m \labs{\norm{\vA_{j}}}^q_{\infty}\R)^{1/q}}{\eta}\labs{\btr\labs{\vR}^p}_q.
\end{align}
The bound for $a_j$ is completely analogous.
Combine the above estimates to obtain the advertised result.
\end{proof}

\subsubsection{Expected moments}\label{sec:resolvent_expected_moment}
To make use of Theorem~\ref{thm:resolvent_con_good_lambda}, we also need to estimate the expected moments via the comparison argument.
\begin{thm}[Expected resolvent moments]\label{thm:expected_resolvent_pq}
For independent centered matrices $\vA_i$, suppose 
the moments match that of idealized matrices $\tvA_i$
\begin{align}
 \Expect  \vA_i= 0 \quad \text{and}\quad\Expect  \vA_i^{\otimes k}  = \Expect  \tvA_i^{\otimes k}  \quad \text{for each}\quad k = 1,\dots,  t \quad \text{and}\quad i = 1,\dots,  m.
\end{align} 
Then,
\begin{align}
    \labs{\labs{\btr \abs{\vR}^p }_q^{1/p}- \labs{\btr \abs{\tvR}^p }_q^{1/p} } 
    \lesssim \frac{1+m/pq}{\eta} \L(\frac{pqL}{\eta}\R)^{t+1} \quad \text{where}\quad L:= \max_i \L( \vertiii{\vA_i}_{3pq}, \vertiiismall{\tvA_i}_{3pq}\R).
\end{align}
The symbol $\lesssim$ suppresses constants depending only on $t$.
\end{thm}

\begin{proof}
We begin with a telescoping sum
\begin{align}
    \btr\abs{\tvR}^p - \btr\labs{\vR}^p&= \sum_{j = 1}^m (\btr \labs{\vR_j}^p - \btr\labs{\vR_{j-1}}^p )=:d_j \quad
\text{with}\quad
    \vR_j := \frac{1}{\vH_j - \omega +\ri \eta}\quad \text{and}\quad  \vH_j := \sum^j_{i=1} \tvA_i +  \sum^m_{j+1} \vA_i.
\end{align}

We move on to control the moments of trace $\labs{\btr\labs{\vR}^p}_q^q $, which uses a similar argument as Theorem~\ref{thm:universality_moments} and Theorem~\ref{thm:universality_resolvent}. We present the calculation for $t=2$, but the general case is analogous. 

\textbf{Step I: Single update error.} We again start with the telescoping sum
\begin{align}
    \Expect \L(\btr\labs{\vR}^p\R)^q - \Expect \L( \btr\abs{\tvR}^p\R)^q &= \sum_{j = 0 }^m \Expect \L(\btr\labs{\vR}^p_j\R)^q - \Expect \L( \btr\abs{\tvR}^p_{j-1}\R)^q.
\end{align}
The Taylor expansions satisfy the bound from Fact~\ref{prop:expand_resolvent}:
\begin{align}
    \btr\labs{\vR_j}^p &= \btr\labs{\vR_-}^p+\sum_{k=1}^{3p} \btr\vM_k \quad \text{where}\quad \labs{\btr\vM_k} \le (4p)^k \normp{\vR_-}{p}^{p-k/3} \frac{\normp{\vA_j}{3p}^k}{\eta^{4k/3}}\\ 
    \L( \btr\labs{\vR_j}^p\R)^q &= (\btr\labs{\vR_-}^p)^q + \sum_{k=1}^{3qp} f_k\quad \text{where}\quad \labs{f_k}\le (4qp)^k (\btr\labs{\vR_-}^p)^{q} \L(\normp{\vR_-}{p}^{-1/3}\frac{\normp{\vA_j}{3p}}{\eta^{4/3}}\R)^k.
\end{align}
The first inequality is analogous with the calculation~\eqref{eq:resolvent_expansion}.
Recall that the $p$-norms are normalized $\norm{\vO}_p = (\btr\labs{\vO}^p )^{1/p}$. The second inequality proceeds with an additional $q^k$ factor.
We then bound the expected increments by canceling the first and second-order terms $f_1$ and $f_2$
\begin{align}
    \labs{\Expect \L(\btr \labs{\vR}^p_j\R)^q - \Expect \L( \btr\abs{\tvR}^p_{j-1}\R)^q} &\le \Expect \sum_{k=3}^{3qp} (4qp)^k \labs{\btr\labs{\vR_-}^p}^{q} \L(\normp{\vR_-}{p}^{-1/3}\frac{\normp{\vA_j}{3p}}{\eta^{4/3}}\R)^k+\L(\vA_j \mapsto \tvA_j\R)\\ 
    &\le \frac{1}{4}\labs{\btr \labs{\vR_-}^p }_q^{q-1/p} \L(\frac{8qp\vertiii{\vA_j}_{3pq}}{\eta^{4/3}}\R)^3
       +\frac{1}{4}\L(\frac{8qp\vertiii{\vA_j}_{3pq}}{\eta^{4/3}}\R)^{3qp}
    +\L(\vA_j \mapsto \tvA_j\R)\\
    &\le \frac{1}{2}\labs{\btr \labs{\vR_-}^p }_q^{q-1/p} \L(\frac{8qp\ell_j}{\eta^{4/3}}\R)^3
    +\frac{1}{2}\L(\frac{8qp\ell_j}{\eta^{4/3}}\R)^{3qp} \\
    &\le \labs{\btr \labs{\vR_j}^p }_q^{q-1/p} \L(\frac{8qp\ell_j}{\eta^{4/3}}\R)^3
    +  \frac{3}{2}\L(\frac{8qp\ell_j}{\eta^{4/3}}\R)^{3qp}.
\end{align}
The second inequality uses Holder's w.r.t the expectation, sums the geometric series, and uses the convenient estimate $\labs{\normp{\vA_j}{3p}}_{3pq} \le \vertiii{\vA_j}_{3pq}$. The third inequality sets $\ell_j:=\max(\vertiiismall{\tvA_j}_{3pq},\vertiii{\vA_j}_{3pq})$. The last inequality establishes a self-bounding argument by 
\begin{align}
    \Expect  \L(\btr \labs{\vR_-}^p \R)^q \le 2\L(\Expect \L( \btr \labs{\vR_j}^p \R)^q +
    \L(\frac{8qp\vertiii{\vA_j}_{3pq}}{\eta^{4/3}}\R)^{3qp} \R)
\end{align}
in a similar vein as~\eqref{eq:resolvent_difference}.

\textbf{Step II: Solving the recursion.} Again, we simplify the recursion by defining scalar variables (and also consider matching moments up to order $t$).
\begin{align}
    x_j:=\Expect (\btr \labs{\vR_j}^p)^q, \quad a_j :=\frac{1}{2^{t-2}} \frac{(2(t+2)qp\ell_j)^{t+1}}{\eta^{t+2}} \quad \text{and}\quad b_j := \frac{3}{2^{t-1}} \frac{(2(t+2)qp\ell_j)^{(t+1)qp}}{\eta^{(t+2)qp}}
\end{align}
and that $a_0 := b_0 :=0$. The updates~\eqref{eq:resolvent_difference} can then be written as a scalar recursion
\begin{align}
     \labs{ x_{j}-x_{j-1}} \le a_j \cdot x_{j-1}^{\frac{qp-1}{qp}} + b_j \quad \text{for each} \quad j = 1 ,\dots,  m 
\end{align}
which, in fact, takes the exact same form as Theorem~\ref{thm:universality_resolvent} up to $p \rightarrow qp$. Regardless, we write down the remaining calculation for completeness. The arguments as before (Lemma~\ref{lem:diff_bounds_discrete}, Lemma~\ref{lem:ansatz}) give the bound for the endpoints
\begin{align}
\labs{x_m^{1/pq} - x_0^{1/pq}} &\le \frac{1}{qp} \sum_{j=1}^m a_j+ (\sum_{j=1}^m b_j)^{\frac{1}{qp}}\\
&\le 8m\frac{(qp)^t ((t+2)L)^{t+1}}{\eta^{t+2}} + (\frac{3}{2})^{1/qp}(1+m/qp) \frac{(2(t+2)qpL)^{t+1}}{2^{(t-1)/qp}\eta^{t+2}} \\
&\le (2^{t+2}+10m/qp) \frac{((t+2) qp L)^{t+1}}{\eta^{t+2}}.
\end{align}
The second inequality uses the uniform bound $\ell_j = \max(\vertiiismall{\tvA_j}_{3p},\vertiii{\vA_j}_{3p}) \le L$. The last inequality drops the denominator $2^{(t-1)/qp} \ge 1$ and uses Young's inequality for $m^{1/qp}\le 1+m/qp$. This is the second advertised result.
\end{proof}

\section{Properties of the GUE}\label{sec:GUE}
In this section, we instantiate the properties of GUE matrices in the nonasymptotic regime. For our purposes, most of the quantities highly concentrate (up $\CO(1/\poly(N))$ deviations) and can be practically regarded as constants.



\begin{defn}[GUE ensemble]
The N-by-N Gaussian Unitary Ensemble is a family of complex Hermitian random matrices specified by
\begin{align}
\vH_{ij} &= \frac{g_{ij}+\ri g'_{ij}}{\sqrt{2N}} \quad \text{if $j>i$}\\
\vH_{ii} &= \frac{g_{ii}}{\sqrt{N}}.    
\end{align}
where $g_{ii}, g_{ij}, g'_{ij}$ are independent standard Gaussians.
\end{defn}

\subsection{$p$-th moments} 

First, we given explicit bounds for the moments of the GUE ensemble.
This kind of result is a consequence of classical explicit formulas for the moments of the GUE;
for example, see~\cite[Lem.~3.3.1]{AGZ10:Introduction-Random}.  It also follows from more recent work on nonasymptotic random matrix theory, such as the paper~\cite{Tro18:Second-Order-Matrix}.
For our purposes, it is convenient for us to derive the statement as a consequence of the main results~\cite[Theorem 2.7]{Bandeira2021MatrixCI}, applied to a GUE matrix.
%


\begin{thm}[Moment bounds]\label{fact:GUE_pth_moment}
 For even $p$ and an random GUE matrix with dimension $N$,
\begin{align}
   \vertiii{\vH_{GUE}}_{p} \le 2 \left(1+ \frac{(p/2)^{3/4}}{\sqrt{N}} \right)
\end{align}
\end{thm}
The number $2$ is exactly the maximal eigenvalue of the semicircular distribution with unit variance.

\begin{proof}
We can compare the following random matrices
    \begin{align}
        \vH_{GUE} = \vX & = \sum_{j\ge i} \L(g_{ij} \frac{\ket{i}\bra{j}+\ket{j}\bra{i}}{\sqrt{2N}} + g'_{ij} \frac{\ri(\ket{i}\bra{j}-\ket{j}\bra{i})}{\sqrt{2N}} \R) \sim \sum_k g_k \vA_k \\
        \vX_{\mathrm{free}} &:= \sum_{k} \vA_k \otimes s_k \label{eq:XXfree}
    \end{align}
    where $(g_{ij}, g_{ij}', g_k)$ are independent standard normal variables and $(s_k)$ composes a free semicircular family. The result~\cite[Theorem 2.7]{Bandeira2021MatrixCI} states that
    \begin{align}
        \labs{ \vertiii{\vX}_{p} - \norm{\vX_{\mathrm{free}}}_{p}} &\le 2(p/2)^{3/4} \sqrt{\sup_{\tr \labs{\vM}^2 =1} \sum_{k}\labs{\tr[\vA_{k}\vM]}^2}\\
        &= 2(p/2)^{3/4} \sqrt{\frac{1}{N}\sup_{\tr \labs{\vM}^2 =1} \sum_{i,j} \labs{M_{ij}}^2} = \frac{2^{1/4}p^{3/4}}{\sqrt{N}}.
    \end{align}
    Recall the unconditional bound $\norm{\vX_{\mathrm{free}}} \le 2 \norm{\BE \vX^2} =2$ to conclude the proof~\cite[p.208]{pisier_2003}.
\end{proof}

\subsection{Resolvent moments}
We calculate the resolvent moments for GUE matrices. Recall 
\begin{align}
    \vR: = \frac{1}{\vH - \omega +\ri \eta} 
\end{align}

\begin{fact}[Consequence of {\cite[Corollary 11.4]{Erds2017ADA}}]\label{lem:GUE_cdf}
For a random instance $\vH_{GUE}$ of the GUE with dimension $N$, define the empirical spectral density
\begin{align}
    \rho(E):= \frac{1}{N} \sum_{i=1}^N \delta(E- \lambda_i(\vH_{GUE}) ).
\end{align}
    Then, there is an absolute constant $c$ such that we have 
\begin{align}
    \sup_{ E \in \mathbb{R} } \labs{\int_{-\infty}^E(\rho(E')-\rho_{sc}(E')) \idiff{E'}} \le \frac{c}{\sqrt{N}}\label{eq:cdf_error}
\end{align}
with probability at least $1-\frac{1}{N}$.
\end{fact}
\begin{proof}
In the setting of~\cite[Corollary 11.4]{Erds2017ADA}, set $D=1$ and $\epsilon = 1/2$. The range $\labs{E}\le 10$ extends to infinity since the semicircle density $\rho_{sc}$ is supported on $[-2,2]$ and the error must be decreasing for $\labs{E} \ge 2$.
\end{proof}

\begin{cor}[Resolvent moments for the GUE]\label{cor:resolvent_moment_GUE}
There is an absolute constant $c$ such that for each $\omega, \eta$, we have
\begin{align}
    \labs{\BE \btr \labs{ \vR_{\omega,\eta}(\vH_{GUE})}^p - S_{\omega,\eta,p}}\le \frac{c}{\eta^p\sqrt{N}} \quad \text{where}\quad S_{\omega,\eta,p} := \int_{-2}^2 \frac{\sqrt{4-x^2}}{2\pi} \frac{1}{\labs{x-\omega + \ri\eta}^p} \idiff{x}. 
\end{align}
\end{cor}
\begin{proof}

Let $f(E) = \labs{E-\omega +\ri \eta}^{-p}$, then
\begin{align}
        \labs{\BE \btr \labs{ \vR_{GUE}}^p - S^p_p}&\le \BE \labs{ \int_{-\infty}^{\infty}f(E)(\rho_{GUE}(E)- \rho_{sc}(E)) \idiff{E} }\\
        &= \BE \labs{ \int_{-\infty}^{\infty} f'(E) \int_{-\infty}^E(\rho_{GUE}(E')- \rho_{sc}(E')) \idiff{E'} \idiff{E} }\\
    &\le \left(1-\frac{1}{N} \right)\cdot\frac{c}{\sqrt{N}}\left( \int_{-\infty}^{\infty} \labs{f'(E)} \diff{E} \right) + \frac{1}{N}\cdot 2\max_{E \in \mathbb{R}} \labs{f(E)}\\
    &\le \frac{2c}{\eta^p\sqrt{N}}+\frac{2}{\eta^pN}.
\end{align}
The first inequality uses integration by parts and the boundary value $\int_{-\infty}^{\infty}(\rho_{GUE}(E')- \rho_{sc}(E')) \diff{E'} =0$. The third line uses Fact~\ref{lem:GUE_cdf} to handle the high probability event~\eqref{eq:cdf_error}.
To reach the last line, we compute the integral using the fact that the resolvent power $f(E) = \labs{E-\omega +\ri \eta}^{-p}$ increasing for $E<\omega$ and decreasing for $E>\omega$, so the integral equals $2f(\omega) = 2 \eta^{-p}$.  To bound the maximum, note that $ 0\le f(E) \le \eta^{-p}$.  Finally, increase the constant $c$ as needed to combine the terms.
\end{proof}

\begin{prop}[GUE: Concentration for resolvent moments]\label{prop:Gaussian_resolvent_concentration}
For a matrix with jointly Gaussian entries
\begin{align}
    \vH:= \sum_i g_i \vA_i \quad \text{where} \quad g_i\sim \CN(0,1)
\end{align}
and even $p$, the spectral density (probed by resolvent powers) concentrates
\begin{align}
\labs{\btr\labs{\vR}^p - \Expect \btr\labs{\vR}^p}_{q} \lesssim \sqrt{q} \frac{p }{\eta^{p+1}} \sigma_* \quad \text{where}\quad \sigma_*:= \sqrt{\sup_{\norm{\bm{w}} = \norm{\bm{v}} = 1} \sum_i \labs{\bra{\bm{v}}\vA_i\ket{\bm{w}}}^2}.    
\end{align}
For a GUE matrix, we have $\sigma_* = N^{-1/2}$.
\end{prop}

\begin{proof}[Proof of Proposition~\ref{prop:Gaussian_resolvent_concentration}]
This is a standard application of Gaussian concentration inequalities. We cannot find the particular function of interest elsewhere, so we include a derivation adapted from~\cite{Tatiana_22_universality}. Some of the estimates could be loose in general but suffice for our purposes.
Consider the function
\begin{align}
    f(\vx) : =\btr\labs{\sum_{i} x_i \vA_i - \omega + \ri \eta}^{-p}.
\end{align}
We bound the Lipschitz constant
\begin{align}
\labs{f(\vx)-f(\vy)} &= \labs{\btr\labs{\vR_{\vx}}^p-\btr\labs{\vR_{\vy}}^p}\\
& = \labs{\btr\L[ \vR_{\vx}^{\dagger}\vR_{\vx}\labs{\vR_{\vx}}^{p-2}  - \vR_{\vy}^{\dagger}\vR_{\vx}\labs{\vR_{\vx}}^{p-2}\R]+ \cdots}\\
&\le \labs{\btr\L[ \vR^{\dagger}_{\vx}(\vH_{\vy}-\vH_{\vx}) \vR_{\vy}^{\dagger}\vR^x\labs{\vR_{\vx}}^{p-2}\R] + \cdots} \le \frac{p }{\eta^{p+1}} \norm{\vH_{\vx} - \vH_{\vy}}. 
\end{align}
The second equality is a telescoping sum. The first inequality uses the identity $\vA^{-1} - \vB^{-1} = \vB^{-1} (\vB-\vA)\vA^{-1}$. The last inequality uses the triangle inequality, the uniform bound that $\norm{\vR} \le \eta^{-1}$, and the coarse bound $\normp{\vA}{1}\le \norm{\vA}$, which holds because we are using normalized Schatten norms.  Last, we relate the operator norm to the Euclidean norm of the coefficients
\begin{align}
    \norm{\vH_{\vx} - \vH_{\vy}} = \lnorm{\sum_i (x_i-y_i)\vA_i}
    & = \sup_{\bm{w},\bm{v}} \sum_i \bra{\bm{v}} (x_i-y_i)\vA_i\ket{\bm{w}}\\
    &\le \left( \sup_{\bm{w},\bm{v}} \sum_i \labs{\bra{\bm{v}}\vA_i\ket{\bm{w}}}^2 \right)^{1/2} \cdot \normp{\vx-\vy}{\ell_2} =\sigma_*\cdot \normp{\vx-\vy}{\ell_2}.
\end{align}
The inequality is Cauchy--Schwarz. Recall that an $L$-Lipschitz function of a
standard Gaussian vector is $L^2$-subgaussian~\cite[Theorem 5.6]{concentration_inequalities_13} to conclude the proof.
\end{proof}

The preceding concentration argument also allows us to bound  
\begin{align}
    \labs{\btr\abs{\tvR}^p}_q &\le \Expect \btr\abs{\tvR}^p+\labs{\btr\abs{\tvR}^p-\Expect \btr\abs{\tvR}^p}_q.
\end{align} 
GUE matrices have strong concentration properties, so the right-hand side of the previous display is always dominated by the expectation term in our applications.

\section{Properties of the Pauli string ensemble }\label{sec:apply_to_pauli}
In the section, we compare properties of random Pauli string sums with the Gaussian Unitary Ensemble (GUE), which we knew a lot about.
Recall 
\begin{align}
\vH_{PS} &= \sum_{j=1}^{m} \vA_j \quad \text{where}\quad \vA_j\stackrel{\text{i.i.d.}}{\sim} \frac{1}{\sqrt{m}} \cdot \pm \{\vI, \vsigma^x,\vsigma^y,\vsigma^z \}^{\otimes n}.
\end{align}
%
We will compare the Pauli string ensemble with the GUE Hamiltonian 
\begin{align}
    \vH_{GUE} = \sum_{j=1}^{m} \tvA_j \quad \text{where}\quad \tvA_j\stackrel{\text{i.i.d.}}{\sim}\frac{1}{\sqrt{m}} \vH_{GUE}.
\end{align}
Comparing the $p$th moments controls the spectral norm, while the resolvent moments control the spectral density.

\subsection{Moments and the spectral norm}\label{sec:pauli_moments}

We use the $p$th moments to bound the spectral norm.  

\pnormsPauli*

To obtain a smaller multiplicative error $\epsilon$, note that the parameters $p$ and $m$ only need to increase at a polynomial rate (as a function of the number $n$ of sites and the parameter $\epsilon$).

\begin{proof}[Proof of Theorem~\ref{thm:pnorms_pauli}]
The $p$-norm estimates use Theorem~\ref{thm:universality_moments}. To obtain the advertised tail bounds, recall the $p$-norm for GUE matrices
\begin{align}
    \vertiii{\vH_{GUE}}_{p} \le 2 \cdot \left( 1+ \frac{(p/2)^{3/4}}{\sqrt{N}} \right).
\end{align}
By Markov's inequality,
\begin{align}
    \Pr\L( \norm{\vH_{PS}} \ge t \R) \le \frac{\Expect \norm{\vH_{PS}}^{p}}{t^p} &\le N\frac{\vertiii{\vH_{PS}}^p_{p}}{t^p}\\
    &\le \L(\frac{\e^{\log(N)/p}}{t} \L(2  + \CO( \frac{p^{3/4}}{\sqrt{N}}+\frac{p^{3/4}}{m^{1/4}} + \frac{p}{ \sqrt{m}} )\R)\R)^{p}\\
    &\le (\frac{1+\epsilon/2 }{1+\epsilon} )^p \le \e^{-c_1 \log(N)/4 } \quad \text{, setting} \quad t= 2(1+\epsilon).
\end{align}
The second inequality converts the operator norm to the $p$-th moments by $\norm{\vH}^p \le \tr\vH^p$ and the third inequality keeps the leading order terms via the notation $\CO(\cdot)$. The third line uses $m \le N^2$ and chooses appropriate parameters $p  = c \log(N)/\epsilon$ and $m = c_1 \frac{p^4}{\log(N)}$ so that the numerator is bounded by $2 (1+\epsilon/2)$. The last inequality uses the elementary estimate $\frac{1+\epsilon/2 }{1+\epsilon} \le \e^{-\epsilon/4}$ for $\epsilon \le 1/2$. Note $N = 2^n$ to obtain the advertised result.

\end{proof}

\subsection{Abundance of low-energy states and success of phase estimation}\label{sec:abundance}

\begin{figure}[t]
    \centering
    \includegraphics[width=0.6\textwidth]{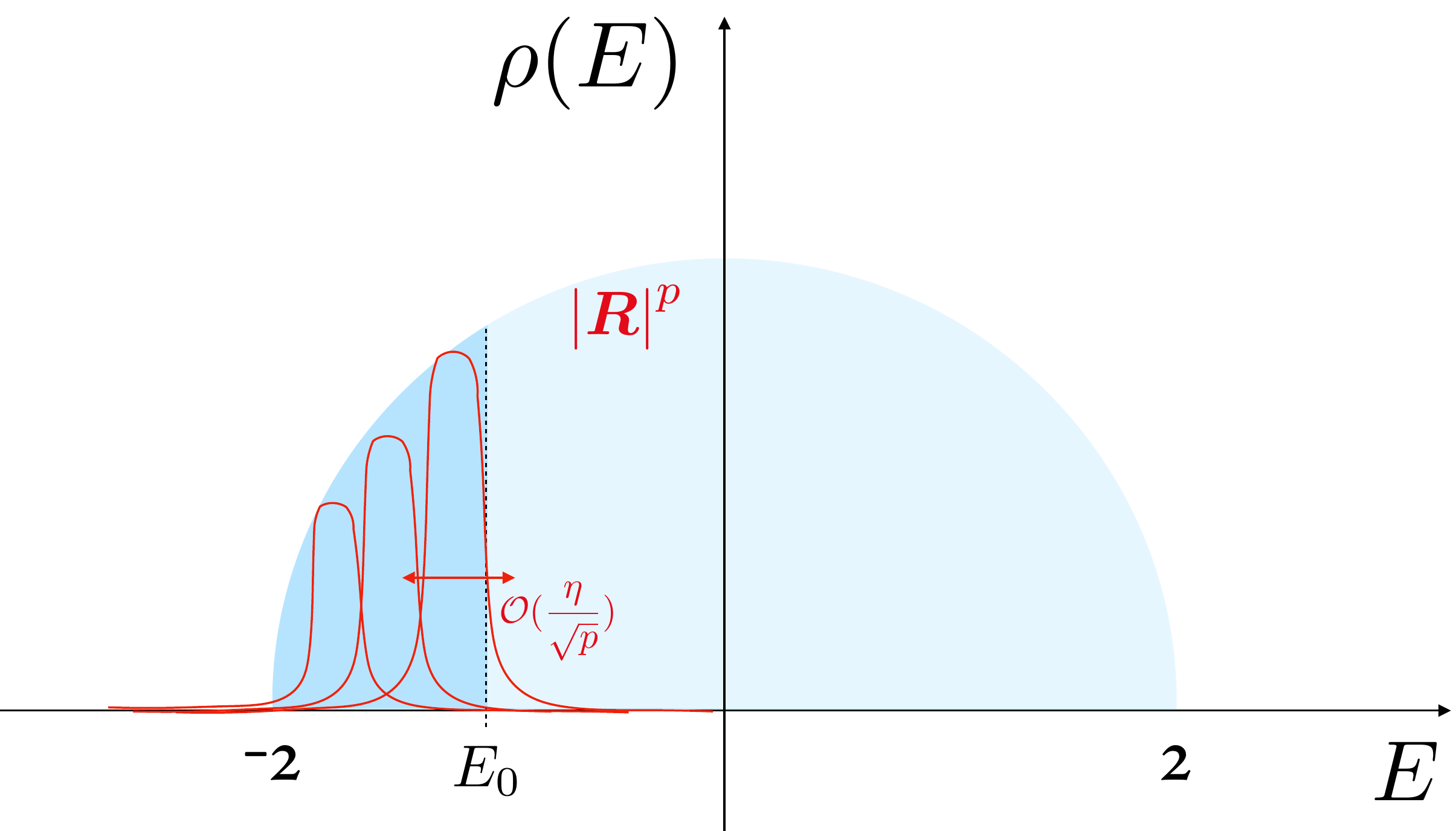}
    \caption{ Probing the low-energy states $E \le E_0$ via consecutive resolvent powers. 
    }
    \label{fig:low_energy_resolvent}
\end{figure}

In this section, we combine the bounds on the minimal eigenvalue (Theorem~\ref{thm:pnorms_pauli}) and the density of states (Theorem~\ref{thm:resolvent_con_good_lambda}) to obtain the low-energy density of states. This immediately implies applying phase estimation on the maximally mixed state returns a low-energy witness with a nonnegligible success probability.
\main*

\begin{proof}[Proof of Theorem~\ref{thm:low-energy-density}]
The resolvent probes the local density of states, and we are interested in controlling the \textit{integrated} density of states in an energy window. The idea is to construct a proxy for the low-energy projector by consecutive local resolvents (Figure~\ref{fig:low_energy_resolvent}). Consider 
\begin{align}
    \eta^p\sum_{2 \le \ell\bomega \le E_0} \labs{\vR_{\ell\bomega,\eta}}^p &=: \sum_{E} \ket{E}\bra{E} q_{E_0}(E). \\
    & =: \vQ(E_0)
    \quad \text{as a proxy for the projector} \quad \sum_{E} \ket{E}\bra{E} \indicator(E \le E_0)
\end{align}
at low-energy $E_0 := -(1-\epsilon/3) \cdot 2$. The resolvents are spaced appropriately
\begin{align}
\labs{\vR_{\ell\bomega,\eta}}^p:=\frac{1}{\labs{\vH - \ell \bomega + i\eta }^p} \quad \text{for} \quad \ell \in \mathbb{Z}, \quad \bomega:= \frac{\epsilon}{\sqrt{p}} \cdot 2, \quad \eta = \frac{\epsilon}{3} \cdot 2, \quad \text{and}\quad p = \lfloor c_1 \log(N)\rfloor.
\end{align}

We will calculate $\Expect  \btr \tvQ(E_0) $ for GUE, show that $\btr{\vQ(E_0)}$ for the Pauli string ensemble takes comparable values, and then extract the low-energy density of state.

\textbf{GUE values.} Recall the GUE resolvent values (Corollary~\ref{cor:resolvent_moment_GUE}) 
\begin{align}
    \Expect \btr\abs{\tvR_{\ell\bomega,\eta}}^p \ge S_p^p/2 \gtrsim \frac{1}{\eta^p}\sqrt{2-(\ell-1) \bomega} \cdot \frac{\epsilon}{\sqrt{p}}
\end{align}
where $S_{\ell\bomega,\eta,p}$ is an integral over the semicircle defined in Corollary~\ref{cor:resolvent_moment_GUE}. The second inequality uses the GUE estimate (Corollary~\ref{cor:resolvent_moment_GUE}) and imposes the simplifying constraint $N \ge const/ \epsilon^4$ such that $S_{\ell\bomega,\eta,p}/2 \ge const. N^{-1/2}$ even near the spectral edge. The third inequality evaluates $S_{\ell\bomega,\eta,p}$.

Apply a crude Riemann sum over the semi-circular density near the edge and drop constants to obtain bounds on the expected value
\begin{align}
    \Expect  \btr \tvQ(E_0) := \sum_{2 \le \ell\bomega \le E_0} \eta^p\Expect \btr\abs{\tvR_{\ell\bomega,\eta}}^p 
    = \Omega\L(\epsilon \sqrt{\epsilon} \R) \quad \text{using} \quad \int \sqrt{x} dx =\frac{2}{3} x\sqrt{x}\label{eq:int_semi-cricle}.
\end{align}

\textbf{Paulis string ensemble values.} Take a crude union bound over the local resolvents, we ensure all of them are at least half of the GUE expectation with high probability
\begin{align}
    \Pr\L(\btr{\vQ(E_0)}\le \frac{1}{2} \Expect  \btr{\tvQ(E_0)} \R) &\le \sum_{\ell \bomega\le E_0} \Pr\L(\btr\labs{\vR_{\ell \bomega,\eta}}^p\le \frac{1}{2} \Expect  \btr\abs{\tvR_{\ell \bomega,\eta}}^p \R) \\
     &\lesssim \sqrt{p} \e^{-\Omega(\log(N)^{1/3}) } \ll 1 \quad \text{(Claim)}\label{eq:Q_concentrion}.
\end{align}

\textbf{Extracting the density of states.} Assuming the claim holds, it remains to extract the spectral density from the low-energy subspace proxy $\vQ(E_0)$. We crudely estimate the function $q_{E_0}(E)$ by spliting the spectrum by a step function 
\begin{align}
    \btr \vQ(E_0) &= \frac{1}{N}\sum_E  q_{E_0}(E)\indicator(E \le E_0+\eta )+ q_{E_0}(E)\indicator(E \ge E_0+\eta )\\
    &\lesssim \frac{1}{N} \sum_E  \indicator(E \le E_0+\eta )+ \e^{-\Omega(p)} \indicator(E \ge E_0+\eta ).
\end{align}
The second inequality uses that $q_{E_0}(E)\lesssim 1$ and that $q_{E_0}(E)=\e^{-\Omega(p)}$ for $E \ge E_0+\eta$. 
Rearrange to bound the number of low-energy states
\begin{align}
    \frac{\# \L\{\ket{E}: E \le -(1-2\epsilon/3)\cdot 2\R\}}{N} = \sum_E\frac{\indicator(E \le E_0+\eta )}{N}    &\gtrsim \eta^p \btr\vQ(E_0) - \e^{-\Omega(p)}\\
    &\gtrsim  \epsilon \sqrt{\epsilon} - \e^{-\Omega(\log(N))} \quad \text{(with high probability~\eqref{eq:Q_concentrion})}\\
    &\gtrsim \epsilon \sqrt{\epsilon} \quad(\text{setting } \log(N) \ge const. \log(1/\epsilon)).
\end{align}
The second inequality uses concentration~\eqref{eq:Q_concentrion} and plugs in the GUE value $\Expect  \btr \tvQ(E_0)$~\eqref{eq:int_semi-cricle}, which is approximately the semicircle integral. The third inequality imposes additional constraints $\log(N) \ge const. \log(1/\epsilon)$, which can be combined with $N \ge const/ \epsilon^4$ by $N \ge \epsilon^{-c_1}$ for some constant $c_1$.  
Combine with the tail bound (Theorem~\ref{thm:pnorms_pauli}) for the operator norm $\norm{\vH_{PS}}$ for $\epsilon'=\epsilon/6$ (using that $\lambda_{\min}(\vH_{PS}) \ge -\norm{\vH_{PS}}$ and that $\frac{1-2/3\epsilon}{1+\epsilon/6} \ge 1-5\epsilon/6$) to obtain
that there are many low-energy states
\begin{align}
    \frac{\# \L\{\ket{E}: E \le (1-5\epsilon/6) \lambda_{\min}(\vH_{PS})\R\}}{N}  \ge \Omega(\epsilon^{3/2}) \quad \text{with high probability}\quad 1-\e^{-\Omega(\log(N)^{1/3})}
\end{align}
drawing from the Hamiltonian ensemble.  Consequently, performing phase estimation with energy resolution $\CO(\epsilon)$ on the maximally mixed state prepares a low-energy witness $\vrho$ such that 
\begin{align}
    \tr[\vrho \vH_{PS}] \le (1-5\epsilon/6)\lambda_{\min}(\vH_{PS}) \quad \text{with success probability}\quad  \Omega(\epsilon^{3/2}).
\end{align}
%
which costs $\poly(m,1/\epsilon)$ using any off-the-shelf quantum simulation algorithm such as Trotter~\cite{lloyd1996universal}, Qubitization~\cite{Low_2019_qubitize}, or qDrift~\cite{campbell2019random} for Hamiltonian simulation within phase estimation. We may amplify the success probability to $1-\epsilon/(6\sqrt{m})$ using $\CO(\epsilon^{-3/2} \log(\sqrt{m}/\epsilon))$ repeats. When all repetitions fail, output the maximally mixed state, which, even in the worst case, has energy upper bounded by $\sqrt{m}$. The resulting output state uses gate complexity
\begin{align}
    G  = (\text{number of repeats})\cdot (\text{QPE cost}) = \Omega(\poly(m,\frac{1}{\epsilon}))
\end{align}
and the energy of the output state\footnote{Strictly speaking, the process we illustrate is a quantum channel involving both quantum gates and classical randomness (i.e., repeating until success is observed). A fixed, deterministic circuit could be constructed by performing phase estimation on half of an input maximally entangled state (for which the reduced density matrix is maximally mixed), and performing fixed-point amplitude amplification \cite{yoder2014fixedpoint} to coherently boost the probability of success.} is at most $(1-\epsilon)\lambda_{\min}(\vH_{PS})$
which is the advertised result.

\textbf{Proof of Claim.} It remains to prove the claim~\eqref{eq:Q_concentrion}; this is where we invoke concentration for the resolvent (Theorem~\ref{thm:resolvent_con_good_lambda}). To reiterate, for each $\vR = \vR_{\ell \bomega, \eta}$, we want to show
\begin{align}
    \text{(WTS)}\quad \Pr\L(\btr\labs{\vR_{\ell \omega,\eta}}^p
    \le \frac{1}{2} \Expect  \btr\abs{\tvR_{\ell \omega,\eta}}^p \R) &\le \e^{-\Omega(\log(N)^{1/3})} \quad \text{for each} \quad \ell, \bomega \label{eq:WTS}
\end{align}
for parameters $\eta = \frac{\epsilon}{3} \cdot 2 $, $p = \lfloor c_1 \log(N)\rfloor$, and
\begin{align}
q &= \theta (\log(N))\\
m &= \Omega(\frac{\log(N)^5}{\epsilon^4}).
\end{align}

For each $\ell,\bomega$, shorthand $\vR = \vR_{\ell \bomega, \eta}$ and rearrange
\begin{align}
    \Pr\L(\btr\labs{\vR}^p
    \le \frac{1}{2} \Expect  \btr\abs{\tvR}^p \R) 
    & \le  \Pr\L(\labs{\btr\labs{\vR}^p - \Expect  \btr\abs{\vR}^p} \ge \frac{1}{2}\Expect  \btr\abs{\tvR}^p - \labs{\Expect  \btr\abs{\tvR}^p -\Expect  \btr\abs{\vR}^p} \R)\\
    &\le \L(\frac{\labs{\btr[\labs{\vR}^p-\labs{\vR'}^{p} ]}_q}{\frac{1}{2}\Expect  \btr\abs{\tvR}^p - \labs{\Expect  \btr\abs{\tvR}^p -\Expect  \btr\abs{\vR}^p}} \R)^{q}.
\end{align}
The last inequality is Markov's. We proceed in bounding the denominator
\begin{align}
    \frac{1}{2}\Expect  \btr\abs{\tvR}^p - \labs{\Expect  \btr\abs{\tvR}^p -\Expect  \btr\abs{\vR}^p} &= \frac{1}{2}\vertiiismall{\tvR}_{p}^p (1 -2 \labs{1 - \frac{\vertiiismall{\vR}_{p}^p}{\vertiiismall{\tvR}_{p}^p}})\\
    &\ge \frac{1}{2}\vertiiismall{\tvR}_{p}^p \L( 1- \CO(\e^{\CO(1/\log(N))}-1)\R).\label{eq:x_RoverR}
\end{align}
The inequality uses $\labs{1-(1+x)^p} \le \labs{1- \e^{\labs{x}p}}$ for
\begin{align}
     x := \labs{\frac{\vertiii{\vR}_{p}}{\vertiiismall{\tvR}_{p}}-1} = \frac{\labs{\vertiii{\vR}_{p}-\vertiiismall{\tvR}_{p}}}{\vertiiismall{\tvR}_{p}} &\lesssim \frac{p^3+p^4/m}{m\eta^4} (1+\frac{p^{3/4}}{N^{1/2}}) =\CO(\frac{1}{p\log(N)}),
\end{align}
which uses comparison of resolvent moments (Theorem~\ref{thm:universality_moments}), that $\vertiii{\vH_{GUE}}_{p} \le 2 \cdot ( 1+ \frac{(p/2)^{3/4}}{\sqrt{N}})$, that $\vertiiismall{\tvR}_{p} \ge 1/2\eta$, and the values of parameters $q,p, m, \eta$.

For the numerator, we evaluate Theorem~\ref{thm:resolvent_con_good_lambda} 
\begin{align}
    \abs{\btr\abs{\vR}^p-\btr\abs{\vR'}^{p}}_q
    &\le \labs{\btr\labs{\vR}^p}_q \L(\frac{\sqrt{q}p^2}{m\eta^2}+ \frac{\sqrt{q}p}{\eta\sqrt{N}}+\frac{qp m^{1/q-1}}{\eta}\R)\\
    &\le \L(\labs{\btr\abs{\tvR}^p}_q^{1/p} + \CO\L( \frac{p^3q^3 + p^4q^4/m }{m\eta^{5}} (1+\frac{(2qp)^{3/4}}{\sqrt{N}})^4\R) \R)^p \cdot \L(\frac{\sqrt{q}p^2}{m\eta^2}+ \frac{\sqrt{q}p}{\eta\sqrt{N}}+\frac{qp m^{1/q-1}}{\eta}\R)\\
    &\lesssim \L(\Expect \btr\abs{\tvR}^p +  \CO( \frac{\log(N)^{1+1/6}}{\eta^{p+1}\sqrt{N}}) \R) \frac{\epsilon^2}{\log(N)^{3-1/6}}.
\end{align}
The first inequality evaluates the second-moment quantity 
\begin{align}
    \sigma_*(\vA_i)^2 = \frac{1}{m}\sup_{\norm{\ket{u}}= \norm{\ket{v}} = 1} \BE_{\vA_i} \labs{\bra{u}\vA_i\ket{v}}^2 = \frac{1}{m} \sup_{\norm{\ket{u}}=\norm{\ket{v}} = 1}\bra{u}\btr[ \ket{v}\bra{v}] \ket{u} = \frac{1}{mN}.
\end{align}

The second inequality compares $\labs{\btr\labs{\vR}^p}_q^{1/p}$ with $\labs{\btr\labs{\tvR}^p}_q^{1/p}$ (Theorem~\ref{thm:expected_resolvent_pq}). 
The last inequality plugs in the values of $q, p, \eta$ in terms of $N,\epsilon$ and uses concentration for Gaussian resolvent $\labs{\btr\abs{\tvR}^p}_q \le \Expect \btr\abs{\tvR}^p +\labs{\btr\abs{\tvR}^p-\Expect \btr\abs{\tvR}^p}_q \le \Expect \btr\abs{\tvR}^p +  \CO( \frac{\sqrt{q} p}{\eta^{p+1}\sqrt{N}})$ (Proposition~\ref{prop:Gaussian_resolvent_concentration}).
We obtain the advertised claim~\eqref{eq:WTS} 
\begin{align}
    \Pr\L(\btr\labs{\vR}^p
    \le \frac{1}{2} \Expect  \btr\abs{\tvR}^p \R) \le \L(\frac{\L(\Expect \btr\abs{\tvR}^p +  \CO( \frac{\log(N)^{1+1/6}}{\eta^{p+1}\sqrt{N}}) \R) \frac{\epsilon^2}{\log(N)^{3-1/6}}}{\frac{1}{2}\vertiiismall{\tvR}_{p}^p (1- \CO(\e^{\CO(1/\log(N))}-1)).}\R)^q &\le \e^{-\Omega(\log(N)^{1/3})}\\
    &\le \e^{-c_3 n^{1/3}}.
\end{align}
The second inequality uses that $\Expect \btr\abs{\tvR}^p = \Expect \btr\abs{\tvR_{\ell\bomega,\eta}}^p \gtrsim \sqrt{\epsilon/\sqrt{p}}^{3}\gtrsim \sqrt{\epsilon/\sqrt{\log(N)}}^{3}$ so that the base is smaller than one for large enough $N$. The last inequality introduces an explicit constant. This concludes the proof of Claim~\eqref{eq:Q_concentrion}. \end{proof}

\section{Circuit size lower bounds for low-energy witness}\label{sec:circuit_lower_bounds}
The GUE is unitarily invariant and thus any subspace of low-energy eigenvectors will be Haar-random. Consequently, preparing a low-energy state of the GUE necessarily requires large circuit complexity. Does enough of this randomness carry over to sparse Hamiltonian ensembles such that a similar statement can be made?
In this section, we prove that with high probability over the ensemble, any circuit that prepares a low-energy witness necessarily has a large circuit size. To show Theorem~\ref{thm:circui_lower_main}, we split into the following two statements. We first calculate the expected norm.
\begin{prop} For the Pauli string ensemble given in Eq.~\eqref{eqn:Pauli-string-ensemble}, 
$
    \Expect  \lambda_{\min}(\vH_{PS}) \le - 1/2.
$
\end{prop}
\begin{proof}
By symmetry of the ensemble,
\begin{align}
   2\Expect \labs{ \lambda_{\min}(\vH_{PS})}= \Expect  \labs{ \lambda_{\min}(\vH_{PS})} +  \Expect  \labs{ \lambda_{\max}(\vH_{PS})} \ge \Expect  \norm{\vH_{PS}} \ge \Expect  \normp{\vH_{PS}}{2} = 1.
\end{align}
The last inequality holds for the normalized Schatten $p$-norms, defined in Eq.~\eqref{eq:define_pnorms}.
\end{proof}

\begin{lem}[small circuit fails to give low-energy states]\label{lem:counting_lower_bounds}
Fix a circuit architecture of two-qubit gates with size $G$ with the initial state $\ket{0}$ and consider the family of all reachable states $\mathrm{Circ}(G)$.
Suppose $m \le \epsilon^2 N^2$, then, there are constants $c_1,c_2$ such that 
\begin{align}
    G \le c_1 \epsilon \sqrt{m} \log^{-1}(m) \quad \text{implies}\quad 
    \Pr\L[ \inf_{\ket{\psi} \in \mathrm{Circ}(G)} \bra{\psi}\vH_{PS}\ket{\psi} \le -\epsilon \R] \le \exp \L( -c_2 \epsilon \sqrt{m} \R). 
\end{align}
\end{lem}
\begin{proof}
The proof uses concentration inequality for any nonrandom state and bootstraps for an epsilon net by the union bound. Consider the random variable associated with a fixed pure state 
\begin{align}
    \braket{\vH_{PS}}_{\psi} = \sum_{i=1}^m a_i \quad \text{where}\quad \braket{\vH_{PS}}_{\psi}:= \bra{\psi} \vH_{PS} \ket{\psi} \quad \text{and} \quad a_i:= \frac{1}{\sqrt{m}}\braket{\vA_i}_{\psi}.
\end{align}
We use the symmetry of the Pauli string ensemble $\vH_{PS} \sim - \vH_{PS}$ to consider the more intuitive maximization problem. We also drop the subscript $\vH = \vH_{PS}$ for simplicity.

\textbf{Variance of a Pauli string.} First, we calculate the variance of a random Pauli string with arbitrary fixed input
\begin{align}
    \BE_{\vA_i} \labs{\braket{\vA_i}_{\psi}}^2 = \BE_{\vA_i} \bra{\psi} \vA_i\ket{\psi}\bra{\psi} \vA_i\ket{\psi} = \frac{1}{N}\bra{\psi}\tr[\ket{\psi}\bra{\psi}]\ket{\psi} = \frac{1}{N} \ll 1.
\end{align}
The second equality evaluates the second moment of the Pauli string ensemble. Notice that the variance is exponentially smaller than the maximal value of the random variable. In other words, a fixed input state is very unlikely to ``align'' with the random Pauli string $\vA_i$. Intuitively, the random Pauli strings are very \textit{noncommutative} and thus cannot be simultaneously diagonalized in a preferred basis. 

\textbf{Variance of the total energy.} 
From the variance of the individual terms, we may obtain a tail bound for the sum via Bernstein's inequality 
\begin{align}
    \Pr[ \braket{\vH}_{\psi}  \ge t] \le \exp\L( \frac{-t^2/2}{v+Lt/3}\R)\quad \text{where}\quad v := \sum_{i}^m \Expect [a_i^2] = \frac{1}{N} \quad \text{and}\quad a_i \le L := 1/\sqrt{m}. \label{eq:one_state_concentration}
\end{align}
In other words, any deterministic input state (that does not correlate with the Hamiltonian) is very likely to have small energies.

\textbf{Union bound over an epsilon net}. By a union-bound, good concentration implies that the energies must be \textit{simultaneously} small for a large family of deterministic input states, specifically, the input states drawn from an epsilon net over a small circuit. For a circuit consisting of $G$ gates, there exists an 
\begin{align}
   \left(\frac{\epsilon}{2\sqrt{m}}\right)\text{-net}\quad \{\ket{\psi_i}\} \quad \text{for} \quad \mathrm{Circ}(G) \quad \text{with cardinality}\quad \#\{\ket{\psi_i}\}  \le \exp\L(\CO( G \log(G\sqrt{m}/\epsilon)) \R)\,.
\end{align}
This is justified as follows. Any circuit with $G$ two-qubit gates is equivalently given by a product of fixed CNOT gates interspersed with $KG$ single-parameter single-qubit rotation gates by certain angles, with $K=\CO(1)$. If we cast a $(\epsilon/(2\sqrt{m}KG))$-net over the interval $[0,2\pi]$ for each of these $KG$ rotation angles, the set of circuits we generate will form an $\epsilon/(2\sqrt{m})$-net over states in $\mathrm{Circ}(G)$, and the cardinality of the set is at $(4\pi\sqrt{m}KG/\epsilon)^{KG}$. One of the elements of this net is guaranteed to approximate the state $\ket{\psi} \in \mathrm{Circ}(G)$ that achieves the supremum of $\braket{\vH}_{\psi}$ up to error $\epsilon/(2\sqrt{m})$, and since $\norm{\vH} \leq \sqrt{m}$ holds, we have that $\sup_{\ket{\psi}} \braket{\vH}_{\psi} \leq \max_{i} \braket{\vH}_{\psi_i} + \epsilon/2$.
We have therefore reduced the supremum over the state on a size-$G$ circuit to the maximum over the $(\epsilon/(2\sqrt{m}))$-net, where the union bound applies~\eqref{eq:one_state_concentration}
\begin{align}
    \Pr\L[ \sup_{\ket{\psi} 
    \in \mathrm{Circ}(G)} \braket{\vH}_{\psi} \ge \epsilon \R] &\le \Pr\L[ \max_{i} \braket{\vH}_{\psi_i} \ge \frac{\epsilon}{2} \R] \\
    &\le \#\{\ket{\psi_i}\} \cdot \exp\L( \frac{-\epsilon^2 /8}{1/N+\epsilon/6\sqrt{m}}\R) \le \#\{\ket{\psi_i}\} \cdot \exp\L(-\min(8\epsilon\sqrt{m}/8,\epsilon^2N/8)\R).
\end{align}
Therefore, there exist constants $c_1,c_2$ such that 
\begin{align}
    G\le c_1 \min( \epsilon \sqrt{m}, \epsilon^2 N )\cdot \log^{-1}(m) \quad \text{implies}\quad 
    \Pr\L[ \sup_{\ket{\psi} \in \mathrm{Circ}(G)} \braket{\vH}_{\psi} \ge \epsilon \R] \le \exp \L( -c_2 \min( \epsilon \sqrt{m}, \epsilon^2 N) \R). 
\end{align}
Plug in the assumption that $m \le \epsilon^2 N^2$ to obtain the advertised result. 
\end{proof}

We suspect the true circuit complexity to be $\Omega(m)$, but the current union bound argument can only give $\tilde{\Omega}(\sqrt{m})$. The concentration inequality needs to handle the event when the same Pauli string occurs $\Omega(\sqrt{m})$-times.

Still, we obtain a growing circuit size lower bound $\Omega(\sqrt{m})$ by an elementary argument. In retrospect, it crucially depends on the \textit{noncommutativity} of Pauli strings: the variance is suppressed by dimension. 
In contrast, the argument only gives $\Omega(n)$ circuit size lower bounds (which is useless) for random complete $k$-local Hamiltonians for fixed $k$. Concretely, let $P_k$ be the set of Pauli strings of weight $k$ and consider the ensemble $\vH = \sum_{ \vsigma\in P_k} r_{\vsigma} \vsigma$ where $r_{\vsigma}$ are uniform random signs. Then,  as in the proof of Lemma \ref{lem:counting_lower_bounds}, define $a_{\vsigma} = r_{\vsigma}\braket{\vH_{\vsigma}}_{\psi}$, and compute (viewing $k = \CO(1)$)
\begin{align}
    v = \sum_{\vsigma \in P_k} \Expect [a_{\vsigma}^2] = \Theta (|P_k|) = \Theta(n^k),\quad \text{and}\quad a_{\vsigma} \le L = 1.
\end{align}
There, the variance is much larger, and the optimum is roughly $\norm{\vH} = \Theta ( \sqrt{v n} )$. Plugging into the union-bound yields
\begin{align}
        \Pr\L[ \sup_{i} \braket{\vH}_{\psi_i} \ge \frac{\epsilon}{2} \sqrt{vn}\R] \le \#\{\ket{\psi_i}\} \cdot \exp\L( \Theta (\frac{-\epsilon^2 vn /8}{v+\epsilon\sqrt{vn}/6} )\R) \sim \#\{\ket{\psi_i}\} \e^{-\Theta ( \epsilon^2 n )}. 
\end{align}
The union bound only supports size-$\CO(n)$ circuits, roughly the circuit size of product states.

\section{Missing proofs}
\label{sec:missing_proofs}
In this section, we collect missing proofs. 
\subsection{Proof of Fact~\ref{fact:complex_sign_match_GUE}}\label{sec:proof_complex_match}
\begin{proof}
The first and third moments vanish for both sets of matrices $\vA_i$ and $\tvA_i$. We calculate the second moment
\begin{align}
    \Expect [\vA_i\otimes \vA_i] &= \frac{1}{2m} \Expect  (\vD\vP + \vP^{\dagger}\vD^{\dagger})\otimes (\vD\vP + \vP^{\dagger}\vD^{\dagger}).\\
    & = \frac{1}{2m} \Expect  \vD\vP\otimes \vP^{\dagger}\vD^{\dagger} +  \vP^{\dagger}\vD^{\dagger}\otimes \vD\vP = \Expect [ \vH_{GUE}\otimes \vH_{GUE} ].
\end{align}
The first equality uses that $\Expect [\vD \otimes \vD^{\dagger}]=0$. The second inequality is that
\begin{align}
    \Expect [ \vD\vP\otimes \vP^{\dagger}\vD^{\dagger} ]&= \sum_i \Expect \ket{i} \bra{i}\vP \otimes \vP^{\dagger}\ket{i} \bra{i} = \frac{1}{N}\sum_i\sum_j \ket{i} \bra{j} \otimes \ket{j} \bra{i}, 
\end{align}
which is proportional to the GUE value. This is the advertised result.
\end{proof}

\section{Difficulty for canceling higher moments via interpolation}
\label{sec:difficult_higher}
In this section, we give a heuristic reason why interpolation-based methods seem difficult to utilize higher-moment-matching. Again, considers a set of independent matrices $\vA_i$ whose low moments match that of some idealized matrices $\tvA_i$
\begin{align}
    \Expect  \vA_i= 0 \quad \text{and}\quad\Expect  \vA_i^{\otimes k} = \Expect  \tvA_i^{\otimes k} \quad \text{for each}\quad k = 1,\dots,  t \quad \text{and}\quad i = 1,\dots,  m. 
\end{align}
Suppose we construct an interpolation path
\begin{align}
    \vS(t) : = f(t)\sum_i \vA_i + g(t)\sum_i \tvA_i\label{eq:f_g_interpolate}
\end{align}
Then, consider the expected $p$-th moment and expand in powers of $\vA_1$ 
\begin{align}
    \Expect \vS(t)^p &= \Expect ( f(t)\vA_1 + g(t)\tvA_1 )\cdots ( f(t)\vA_1 + g(t)\tvA_1 )\cdots \\
    & + \Expect( f(t)\vA_1 + g(t)\tvA_1 )\cdots ( f(t)\vA_1 + g(t)\tvA_1 )\cdots ( f(t)\vA_1 + g(t)\tvA_1 )+\cdots .
\end{align}
Suppose the second and third moments do not vanish. If we wish the time-derivative to vanish, then we generally need
\begin{align}
    f(t)^2\Expect \vA_1\otimes \vA_1 + g(t)^2 \Expect \tvA_1\otimes \tvA_1 = Const. \implies 
    f(t)^2 + g(t)^2 =1 \\
    f(t)^3\Expect \vA_1\otimes \vA_1\otimes \vA_1 + g(t)^3 \Expect \tvA_1\otimes \tvA_1\otimes \tvA_1 = Const. \implies 
    f(t)^3 + g(t)^3 =1.
\end{align}
There are only discrete solutions to both algebraic equations and no continuous path can be established between $f = 0$ and $f =1$. Indeed, the standard interpolant $f(t) = \sqrt{1-t}, g(t)=\sqrt{t}$ only cancels out the second moments. Therefore, if we hope interpolation methods capture higher moments matching conditions, we need to go beyond the form of~\eqref{eq:f_g_interpolate}. On the contrary, the Lindeberg principle appears more natural for this task. 

\end{document}